\documentclass[11pt]{article}

\usepackage{fullpage}

\newcommand{\ER}{Erd\H{o}s-R\'{e}nyi }
\newcommand{\A}{\mathcal{A}}

\usepackage{comment}
\usepackage{epsf}
\usepackage{graphicx}
\usepackage{subfigure}
\usepackage{bbm}
\usepackage{color}
\usepackage{nicefrac}
\usepackage{mathtools}
\usepackage{amsfonts}
\usepackage{amsthm}
\usepackage{amsmath, bm}
\usepackage{amssymb}
\usepackage{textcomp}
\usepackage{xcolor}
\usepackage[ruled,vlined]{algorithm2e}

\usepackage{algorithm}   
\usepackage{algpseudocode}
\usepackage{enumitem}
\makeatletter
\algnewcommand{\LineComment}[1]{\Statex \hskip\ALG@thistlm {\color{gray}\texttt{// #1}}}
\makeatother

\usepackage{hyperref}
\definecolor{OliveGreen}{HTML}{3C8031}
\hypersetup{
colorlinks=true,
urlcolor=blue,
linkcolor=blue,
citecolor=OliveGreen,
linktoc=page,
}

\newenvironment{fminipage}%
  {\end{minipage}\end{Sbox}\fbox{\TheSbox}}

\allowdisplaybreaks

\newenvironment{claimproof}{\noindent$\rhd$\hspace{1em}}{\hfill$\lhd$}



\newtheorem{theorem}{Theorem}[section]

\newtheorem{proposition}[theorem]{Proposition}
\newtheorem{lemma}[theorem]{Lemma}

\newtheorem{definition}[theorem]{Definition}
\newtheorem{remark}[theorem]{Remark}

\newtheorem{question}[theorem]{Question}
\newtheorem{conjecture}[theorem]{Conjecture}
\newtheorem{claim}[theorem]{Claim}

\counterwithin{algorithm}{section}

\usepackage[backend=biber, style=numeric, sorting=nyt, maxnames=100,backref=true, sortcites = true]{biblatex}
\newlength{\widebarargwidth}
\newlength{\widebarargheight}
\newlength{\widebarargdepth}

\makeatletter
\long\def\@makecaption#1#2{
       \vskip 0.8ex
       \setbox\@tempboxa\hbox{\small {\bf #1:} #2}
       \parindent 1.5em 
       \dimen0=\hsize
       \advance\dimen0 by -3em
       \ifdim \wd\@tempboxa >\dimen0
               \hbox to \hsize{
                       \parindent 0em
                       \hfil 
                       \parbox{\dimen0}{\def\baselinestretch{0.96}\small
                               {\bf #1.} #2
                               } 
                       \hfil}
       \else \hbox to \hsize{\hfil \box\@tempboxa \hfil}
       \fi
       }
\makeatother

\long\def\comment#1{}

\DeclareMathOperator{\epst}{\Tilde{\varepsilon}}
\DeclareMathOperator{\eps}{\varepsilon}

\newcommand{\ind}[1]{\ensuremath{{\mathbb{I}\left\{ #1 \right\}}}}

\newcommand{\set}[1]{\{#1\}}

\newcommand{\N}{\mathbb{N}}

\newcommand{\E}{\mathbb{E}}
\renewcommand{\P}{\mathbb{P}}
\newcommand{\Q}{\mathbb{Q}}

\newcommand{\statunif}{\alpha_{\texttt{r}}^{\rm STAT}}
\newcommand{\statpartite}{\alpha_{\texttt{r-bal}}^{\rm STAT}}

\newcommand{\compunif}{\alpha_{\texttt{r}}^{\rm COMP}}
\newcommand{\comppartite}{\alpha_{\texttt{r-bal}}^{\rm COMP}}

\title{Algorithmic Phase Transition for Large Independent Sets in Dense Hypergraphs}

\usepackage{authblk}

\author[1]{Abhishek Dhawan\thanks{Email: \textit{adhawan2@illinois.edu}. Partially supported by the NSF RTG grant DMS-1937241.}}

\author[2]{Nhi U. Dinh\thanks{Email: \textit{nhidinh2@illinois.edu}.}}

\author[3]{Eren C. K{\i}z{\i}lda\u g\thanks{Email: \textit{kizildag@illinois.edu}.}}

\author[4]{Neeladri Maitra\thanks{Email: \textit{nmaitra@illinois.edu}.}}

\author[5]{Bayram A. \c{S}ahin\thanks{Email: \textit{basahin2@illinois.edu}.}}

\affil[1,2,4]{Department of Mathematics, University of Illinois Urbana-Champaign}

\affil[3,5]{Department of Statistics, University of Illinois Urbana-Champaign}

\date{}

\begin{document}
\maketitle

\begin{abstract}
    We study the algorithmic tractability of finding large independent sets in dense random hypergraphs. In the sparse regime, much of the natural algorithms can be formulated within either the local or the low-degree polynomial (LDP) framework,
    and a rich literature has subsequently identified nearly sharp algorithmic thresholds within these classes by exploiting their stability. In the dense setting, however, the algorithmic paradigms are fundamentally different: they are online and thus need not be stable. Perhaps more crucially, even for the classical \ER random graph $G(n,p)$, LDPs are conjectured to fail in the `easy’ regime accessible to online algorithms, thereby challenging their viability for dense models.
    
Our focus is on two models: (i) finding large independent sets in dense $r$-uniform \ER hypergraphs, where each size-$r$ hyperedge is present independently with probability $p$, and (ii) the more challenging problem of finding large $\gamma$-balanced independent sets in dense $r$-uniform $r$-partite hypergraphs, where the vertex set is the disjoint union $V_1\sqcup \cdots \sqcup V_r$ with $|V_i|=n$ for all $i$, each hyperedge in $V_1\times \cdots \times V_r$ is present independently with probability $p$, and the $i$-th coordinate of $\gamma\in\mathbb{Q}^r$ specifies the proportion of vertices from $V_i$ in the independent set. For both models, we pinpoint the size of the largest independent set 
 and design online algorithms that achieve a multiplicative approximation factor of $r^{1/(r-1)}$ in the uniform and  $(\max_i \gamma_i)^{-1/(r-1)}$ in the $r$-partite model. Furthermore, we establish matching algorithmic lower bounds, showing that these computational gaps are \emph{sharp}: no online algorithms can breach these gaps.

Our results provide a detailed landscape for dense hypergraphs, thereby completing the picture for dense models in a manner parallel to the sparse counterparts developed recently. Our main technical contribution is twofold: a novel \emph{staged and bucketed greedy} algorithm and a stopping-time argument tailored to the hypergraph and multipartite structure for algorithmic hardness, both of which may be of independent interest. The algorithms and proof techniques in the dense regime differ substantially from those in the sparse, yet the resulting computational gaps are remarkably analogous, pointing to a form of universality. More conceptually, our results corroborate a hypothesis from statistical mechanics linking glassy equilibrium to computational hardness: optimization problems become far more intricate in the presence of global constraints.
\end{abstract}

\pagenumbering{gobble}
\newpage
\tableofcontents
\newpage
\pagenumbering{arabic}

\sloppy

\section{Introduction}

In this paper, we study the statistical and algorithmic landscape of large independent sets in two hypergraph models: (i) the $r$-uniform \ER hypergraph $H_r(n,p)$, where each size-$r$ subset of $[n]\coloneqq \{1,\dots,n\}$ is present independently with probability $p$, and (ii) the $r$-uniform $r$-partite hypergraph $H(r,n,p)$, whose vertex set is partitioned into $r$ parts $V_1,\dots,V_r$, each of size $n$,  where each edge in  $V_1\times \cdots \times V_r$ is present independently with probability $p$. For formal statements, see Definition~\ref{definition: models}. Our focus is on the \emph{dense} regime in which $p$ is a constant (as $n\to\infty$). While the landscape for sparse random hypergraphs has been characterized in detail recently~\cite{dhawan2026low}, the dense regime has remained elusive and requires different techniques. We examine large independent sets in $H_r(n,p)$ and large \emph{balanced} independent sets in $H(r,n,p)$, where the selected vertices are required to be distributed across the $r$ parts according to a prescribed vector $\gamma \in \Q_+^r$; see Definition~\ref{definition: balanced independent sets}. For both models, we pinpoint the largest possible size—referred to as the \emph{statistical threshold}, design polynomial-time algorithms that provably outputs a large independent set, and establish sharp algorithmic lower bounds. Taken together, our results reveal \emph{statistical-computational gaps} analogous to those in the sparse regime,
thereby pointing to a form of universality. 

Establishing sharp results for dense hypergraphs requires new technical ingredients. To this end, we introduce a novel \emph{staged and bucketed greedy} algorithm tailored to hypergraphs and design a geometrical barrier, equipped with a carefully chosen stopping time, to establish tight algorithmic lower bounds. Informally, the \emph{staged and bucketed} greedy procedure is divided into $r$ stages where in each stage, the algorithm populates the independent set from an `unexhausted' part until a certain target capacity is reached, and then locks that part. As for our algorithmic barrier, it is adapted to the multipartite structure and global balancedness constraint; specifically, it  tracks the first time at which the independent set contains sufficiently many vertices from every part of the partition. See  Section~\ref{sec:pf-overview} for a more detailed overview. We expect these new tools to be of independent interest for understanding algorithmic tractability in hypergraph models and beyond.

Our algorithmic viewpoint is \emph{online}: vertices are exposed sequentially, and upon the arrival of a vertex $v$, the algorithm learns the status of hyperedges involving $v$ together with previously exposed vertices, and must irrevocably decide whether to include $v$ in its output; see Definition~\ref{def : online-arrival} for details. As we explain below, online algorithms are central in the dense setting, where they capture the {only known class of} polynomial-time algorithms. This stands in stark contrast to the sparse setting, where the predominant algorithmic frameworks are markedly different, see Section~\ref{sec:compare}.

We emphasize that, while approximating the largest independent set to within a factor of $n^{1-\epsilon}$ is already NP-hard in the worst-case~\cite{hastad-clique,khot2001improved}, the average-case picture is substantially richer. For the \ER random graph $G(n,\frac12)$, the independence number—size of the largest independent set—is approximately $2\log_2 n$, whereas the online greedy algorithm reaches $\log_2 n$, both with high probability (whp) as $n\to\infty$~\cite{bollobas1976cliques,Grimmett_McDiarmid_1975,matula1970complete,matula1976largest}. In his seminal 1976 work~\cite{karp1976probabilistic}, Karp challenged the community to design an efficient algorithm that reaches $(1+\epsilon)\log_2 n$ (whp) for some $\epsilon>0$, or to prove that this is impossible (modulo a complexity-theoretic assumption such as $P\ne NP$). Nearly five decades later, this question remains open and is widely regarded as a cornerstone of average-case complexity and random graph theory. It has inspired several major developments—such as Jerrum’s introduction of the \emph{planted clique} problem~\cite{jerrum1992large}, and a substantial literature has since uncovered analogous `factor-$2$ gaps' across a broad range of random graphs models.\footnote{Importantly, Karp's original conjecture regarding `factor-$2$ gaps' was formulated for the dense regime, while much of the recent effort has been focused on the sparse setting.} For an overview of such gaps in random graphs and beyond, see the surveys~\cite{gamarnik2021overlap,gamarnik2025turing,gamarnik2022disordered,bandeira2018notes,wein2025computational,frieze2005random}.

\subsection{Comparison between Online Algorithms and 
Low-Degree Polynomials}\label{sec:compare}
There is by now a substantial literature on algorithmic phase transitions for the largest independent set problem in sparse random graphs (including the classical \ER model as well as bipartite and multipartite analogues), where nearly sharp lower bounds have been established for local algorithms and low-degree polynomials (LDPs)~\cite{gamarnik2014limits,rahman2017local,wein2020optimal,perkins2024hardness,dhawan2026low,gamarnik2020low}. In the sparse regime, this framework is particularly natural: local neighborhoods remain small and approximately tree-like, and several natural algorithms can be formulated within either the local or low-degree framework~\cite{wein2020optimal,perkins2024hardness,dhawan2026low,gamarnik2020low}.

\begin{remark}
    We note that in the low-degree setting, the local tree-like property is essential to control the variance of certain quantities. We discuss this further in relation to our work below.
\end{remark}

The dense regime, however, is fundamentally different. Even for the classical \ER graph $G(n,p)$ with constant $p$, it is likely that the online greedy algorithm cannot be implemented as an LDP. An informal justification for this is through a certain \emph{invariant}: while LDPs enjoy input stability~\cite{wein2020optimal,gamarnik2020low,perkins2024hardness}, the online greedy algorithm is shown to be unstable~\cite{gamarnik2025optimal}, thus unlikely to be implemented as an LDP. The current state of affairs is in fact more subtle. A remarkable conjecture put forth in a recent AIM workshop~\cite{AIM2024} suggests that LDPs fail even in the `easy' regime:
\begin{conjecture}\label{conj:LDP-fails}
In $G(n,\frac12)$, no degree-$o(\log^2 n)$ polynomial can find an independent set of size $0.9\log_2 n$. 
\end{conjecture}
Accordingly, while LDPs provide an important proxy for efficient algorithms in sparse random graphs, their viability in the dense setting, even for the classical \ER random graph $G(n,p)$ with constant $p$, is questionable.
\footnote{Degree-$o(\log^2 n)$ polynomials are used as a proxy for polynomial-time algorithms; see~\cite{hopkins2018statistical,kunisky2019notes,wein2025computational}.} If true, Conjecture~\ref{conj:LDP-fails} would show that LDPs are actually not viable for dense random graphs. A relatively weaker variant of this conjecture was proven in the aforementioned workshop~\cite{AIM2024}. This motivates the study of online algorithms in the dense regime.

\begin{remark}
The online greedy algorithm is in fact relevant well beyond the dense setting. Indeed, in the intermediate regime where $p=o(1)$ and $np\to\infty$, the independence number of $G(n,p)$ is $(2+o(1))\log(np)/p$, whereas the greedy algorithm finds an independent set of size $(1+o(1))\log(np)/p$, both whp~\cite{Bollobas2001book,frieze1990independence,frieze1997algorithmic}. The same phenomenon persists in the sparse regime: in the double limit $n\to\infty$ followed by $d\to\infty$,  the independence number of $G(n,d/n)$ is $(2+o_d(1))n\log d/d$, while greedy reaches $(1+o_d(1))n\log d/d$~\cite{coja2015independent}. Thus, throughout dense, intermediate, and sparse regimes, the online greedy yields $1/2$-approximation to the optimum (at the level of the leading-order terms).
\end{remark}

Lower bounds for computational intractability of finding large independent sets in $G(n,p)$ within the framework of online algorithms had remained elusive, emerging only very recently~\cite{gamarnik2025optimal}.\footnote{We note, however, that lower bounds for LDPs appeared earlier in~\cite{huang2025strong,wein2020optimal}.} At a technical level, they required substantial  refinements of the machinery based on the Overlap Gap Property (see below for details). These results were recently extended to the balanced independent set problem in dense bipartite random graphs~\cite{dhawan2025sharp}. As we detail below, the largest balanced independent set problem in bipartite models closely resemble their classical \ER counterparts (even though the unconstrained independent set problem in bipartite random graphs is polynomial-time solvable via a max-flow reduction). As such, a variant of Conjecture~\ref{conj:LDP-fails} plausibly holds for bipartite models—this was in fact put forth in~\cite[Conjecture~1.6]{dhawan2025sharp}.

\paragraph{Bipartite Random Graphs} For the large balanced independent set problem in sparse random bipartite graphs, Perkins and Wang~\cite{perkins2024hardness} obtained sharp algorithmic results within both the local and LDP frameworks. Their algorithm arises from bipartiteness, but the sparsity is crucial in placing it into the local or LDP framework—their analysis breaks down in the absence of sparsity. They explicitly note that, since sparse random graphs are locally tree-like, the performance of local algorithms is determined, to first order, by evaluation on a Galton-Watson tree. The same is true also in the bipartite case. More specifically, the analysis of the local algorithm uses the fact that for sparse random graphs, the neighborhoods are asymptotically tree-like and hence can be approximated by a Galton-Watson tree, which they leverage to control certain acceptance probabilities and to bound neighborhood sizes; see~\cite[Section~2.1]{perkins2024hardness}. As for the low-degree analysis, one must control the variance of a certain LDP, for which it is essential to control the size of certain neighborhoods. In~\cite[Section~2.2]{perkins2024hardness}, this is handled precisely through a Poisson approximation, which unfortunately holds only in the sparse regime and is no longer valid for the dense setting.

For these reasons, results of~\cite{perkins2024hardness} for sparse bipartite graphs do not extend to the dense setting. In a dense bipartite graph, even a depth-1 neighborhood already contains $\Theta(n)$ vertices (along with substantial dependencies), so the local limit argument based on Galton-Watson approximation disappears. Consequently, if one tries to adapt the algorithmic framework of~\cite{perkins2024hardness}, the resulting polynomial will unfortunately have a large variance, breaking down the analysis. As for the balanced independent set problem in sparse random multipartite hypergraphs, recent work of Dhawan and Wang~\cite{dhawan2026low} devised LDP algorithms, for which the analysis does not transfer to dense regime either.

For these models, the main source of hardness is not multipartiteness per se, but rather the \emph{balancedness} constraint. More broadly, once one imposes a global constraint such as balancedness, the resulting optimization problem becomes substantially more delicate.\footnote{In the language of statistical mechanics, a global constraint such as balancedness is expected to induce `glassy equilibrium’ and lead to computational hardness~\cite{mezard1987mean,mezard2010glassy,perkins2024hardness,gamarnik2022disordered}; this has been rigorously verified for certain models, see, e.g.,~\cite{feige2002relations,feige2004hardness}.
}
In particular, the largest balanced independent set problem in dense bipartite or multipartite random graphs is much closer in flavor to the classical largest independent set problem in dense $G(n,p)$ than to the unconstrained independent set problem on bipartite graphs. In light of this, online algorithms offer a natural algorithmic framework in the dense regime, across both the \ER and multipartite settings, whereas the viability of the low-degree framework (or local algorithms) is presently much less clear.\footnote{For related reasons, online algorithms also play an important role in other optimization problems such as coloring over dense random graphs; see~\cite{krivelevich2002coloring}.} 
\vspace{-0.3in}
\paragraph{Extension to Hypergraphs and Multipartite Models} The higher uniformity ($r\ge 3$), the multipartite structure, and the balancedness constraint collectively make the rigorous study of random hypergraph models substantially more delicate. To address these challenges, we develop new technical ideas both on the algorithmic side and for establishing sharp algorithmic hardness.
 These ideas contribute to the technical toolkit for dense random structures, which has only recently begun to emerge; we expect them to be broadly useful for hypergraph models and beyond. 

\begin{remark}
While it is known that a large class of pseudorandom graphs, such as $K_3$-free graphs, have independence number matching the computational threshold (see, e.g.,~\cite{DKPS, shearer1983note, Kttt}), this is not known for \textbf{any} class of pseudorandom hypergraphs. There are, however, several results that match the growth rate~\cite{ajtai1982extremal, frieze2013coloring, dhawan2026fractional}.
These results suggest that  finding large independent sets in hypergraphs is inherently harder than the graph case,
further motivating the study of hypergraphs.
\end{remark}
For a detailed technical overview, see Section~\ref{sec:pf-overview}. A summary of our contributions is in order.

\subsection{Summary of Main Results}\label{subsection: informal results}

We now provide informal versions of our main results; see Section~\ref{sec: main results} for the formal statements. 
Recall the informal descriptions of our models: the first is the dense $r$-uniform Erd\H{o}s--R\'enyi hypergraph $H_r(n,p)$, where each $r$-element subset of $[n]\coloneqq \{1,\dots,n\}$ is present independently with probability $p$; the second is the dense $r$-uniform $r$-partite hypergraph $H(r,n,p)$, where the vertex set is partitioned into $r$ parts $V_1,\dots,V_r$ of size $n$ and each $r$-tuple from $V_1\times \cdots \times V_r$ is included independently with probability $p$ (see Definition~\ref{definition: models} for a formal definition of these models).

For $b\coloneqq \frac{1}{1-p}$ and $\gamma\coloneqq (\gamma_1,\dots,\gamma_r)\in\Q_{+}^r$
with $\sum_i \gamma_i=1$, define
\[\statunif \coloneqq \left(r!\log_b n\right)^{\frac{1}{r-1}}\qquad\text{and}\qquad \statpartite\coloneqq \left(\frac{\log_b n}{\prod_{i=1}^r \gamma_i}\right)^{\frac{1}{r-1}}.\]
Throughout this paper, we focus on constant $p$, $p=\Theta(1)$ (as $n\to\infty$). We ignore all floor/ceiling operators for simplicity with the understanding that this has no effect on our arguments.

Our first result pins down the size of the largest ($\gamma$-balanced) independent set. 
\begin{theorem}[Informal version of Theorems~\ref{theo: a_stat-Hr} and~\ref{theo: gamma_bal a_stat}]\label{thm:stat-threshold-informal}
For $H_r(n,p)$, the largest independent set  has size $\left(1\pm o(1)\right)\statunif$ whp. Likewise, for $H(r,n,p)$, the largest $\gamma$-balanced independent set has size $\left(1\pm o(1)\right)\statpartite$ whp.
\end{theorem}

Theorem~\ref{thm:stat-threshold-informal} identifies $\statunif$ and $\statpartite$ as the \emph{statistical threshold} for independent sets in $H_r(n,p)$ and $\gamma$-balanced independent sets in $H(r,n,p)$, respectively. The proof of Theorem~\ref{thm:stat-threshold-informal} is based on the second moment method and is, therefore, non-constructive. This naturally leads to an algorithmic question:
can we efficiently find large independent sets in $H_r(n,p)$ or large $\gamma$-balanced independent sets in $H(r,n,p)$? 
To address this question, we let
\[\compunif\coloneqq \left((r-1)!\log_b n\right)^{\frac{1}{r-1}}\qquad\text{and}\qquad \comppartite\coloneqq \left(\frac{\log_b n}{\prod_{i=1}^r \gamma_i}\max_{1\le i\le r}\gamma_i\right)^{\frac{1}{r-1}}.\]
Our next result establishes that the thresholds $\comppartite$ and $\compunif$ are achievable in polynomial time via an online algorithm.
\begin{theorem}[Informal version of Theorems~\ref{theo: a_COMP-achievable-Hr} and~\ref{theo: a_comp-achievable-Hrnp}]\label{thm:alg-informal}
For $H_r(n,p)$, there is an online algorithm which finds, for any $\epsilon>0$, an independent set of size $(1-\epsilon)\compunif$ whp. Likewise, for $H(r,n,p)$, there is an online algorithm which finds, for any $\epsilon>0$, a $\gamma$-balanced independent set of size $(1-\epsilon)\comppartite$ whp.
\end{theorem}
We now compare the statistical thresholds with the algorithmically achievable values. For the dense $r$-uniform hypergraph $H_r(n,p)$, there is a factor-$r^{1/(r-1)}$ gap between the two, which recovers the well-known `factor-$2$ gap' for the special case $r=2$. As for $H(r,n,p)$, the gap is of the order $\left(\max_{1\le i\le r}\gamma_i\right)^{-1/(r-1)}$. (Note that for $\gamma_1=\cdots=\gamma_r=1/r$, the gaps in both cases are identical.)

Are there fundamental computational barriers in these models? Our final result addresses this question and yields sharp algorithmic lower bounds, suggesting that the barriers above are not spurious and likely inherent.
\begin{theorem}[Informal version of Theorems~\ref{theo: a_COMP-impossible-Hr} and~\ref{theo: a_COMP-impossible-Hrnp}]\label{thm:ogp-informal}
    For $H_r(n,p)$ and any $\epsilon>0$, no online algorithm finds an independent set of size $(1+\epsilon)\compunif$ with probability $o(1)$. Likewise, for $H(r,n,p)$ and any $\epsilon>0$, no online algorithm finds a $\gamma$-balanced independent set of size $(1+\epsilon)\comppartite$ with probability $o(1)$.
\end{theorem}

Note that aside from onlineness, there is no restriction on the algorithms ruled out by Theorem~\ref{thm:ogp-informal}. In particular, so long as the algorithm is online, we impose no restrictions on its runtime. Importantly, Theorem~\ref{thm:ogp-informal} rules out algorithms succeeding even with a vanishing probability $o(1)$. In relevant literature, this is known as \emph{strong hardness}; see~\cite{huang2025strong} as well as~\cite{gamarnik2025optimal,dhawan2025sharp}. 
\begin{remark}
    More explicitly, the $o(1)$ guarantee is of the form $\exp(-\Omega(\compunif\log n)$ for $H_r(n,p)$ and of form $\exp(-\Omega(\comppartite\log{n}))$ for $H(r,n,p)$. Both of these guarantees are essentially optimal, since the probability that a randomly chosen subset with size reaching the computational threshold is independent in $H_r(n,p)$ (or a $\gamma$-balanced independent set in $H(r,n,p)$) is of the same order.
    \end{remark}

Taken together, Theorems~\ref{thm:alg-informal} and~\ref{thm:ogp-informal} sharply characterize the performance of online algorithms in dense $r$-uniform hypergraphs and dense $r$-partite hypergraphs, suggesting an inherent computational barrier, precisely at $\compunif$ and $\comppartite$. Our result recovers the analogous computational gaps present for $r=2$, generalizing them to $r>2$, thus pointing towards a certain universality.\footnote{We remark that such a universality arises also in the context of sparse random graphs, see~\cite{perkins2024hardness,dhawan2026low,wein2020optimal}.}

\begin{remark}
Note that for $H_r(n,p)$, the computational gap is of the order $r^{1/(r-1)}$, which vanishes as $r\to\infty$. Likewise, the same conclusion also holds for $H(r,n,p)$, using the fact $\max_i \gamma_i\ge 1/r$. Thus, the computational problem gets `easier' when the sizes of the hyperedges become large.
\end{remark}

\begin{remark}
As discussed below, we introduce a novel variant of the greedy algorithm—staged and bucketed greedy—together with a carefully designed stopping time that accounts for both the multipartite structure and the balancedness constraint. Both ingredients are crucial for addressing hypergraphs and multipartite structures, and for moving beyond the classical $G(n,p)$ and bipartite setting. In this sense, our results substantially generalize several existing results in the literature. Indeed, for $H_r(n,p)$ with $r=2$, we recover both Karp's algorithmic result~\cite{karp1976probabilistic} as well as the lower bound of~\cite{gamarnik2025optimal}. Likewise, for $H(r,n,p)$, our results with $r=2$ recover the full set of results  in~\cite{dhawan2025sharp}.
\end{remark}

\subsection{Proof Overview}\label{sec:pf-overview}

Our proofs for $H_r(n,p)$ and $H(r,n,p)$ follow a consistent three-part strategy: identifying the statistical threshold, demonstrating achievability via online algorithms, and applying OGP to establish matching algorithmic lower bounds. We focus first on $H_r(n,p)$ to clarify the underlying mechanics. We then discuss the technical innovations required for the more complex $H(r,n,p)$ setting, including the \emph{staged and bucketed} greedy procedure and a new stopping time argument, both of which account for the multipartite structure and the global balancedness requirement. 

The statistical threshold is established via the second moment method~\cite{alon2016probabilistic}. As for the algorithms, our  approach is through an online greedy algorithm. To the best of our knowledge, ours is the first work to analyze the greedy independent set procedure on \ER hypergraphs for $r \geq 3$.

Our impossibility result is positioned in a body of work initiated by Gamarnik and Sudan~\cite{gamarnik2014limits}, who investigated local algorithms for independent sets in sparse random graphs using the Overlap Gap Property (OGP); see below for a background. This property yields a dichotomy: the intersection of two ``large" independent sets must be either significantly large or very small, creating a forbidden ``gap'' in overlap sizes. They demonstrated that a successful local algorithm could be leveraged to produce independent sets that violate this OGP. While this approach was refined by Rahman and Vir\'ag~\cite{rahman2017local} and extended by Wein~\cite{wein2020optimal} to LDPs, these methods fail in dense settings. For dense random graphs, LDPs—and local algorithms, which can be implemented as LDPs—are conjectured to fail; see Conjecture~\ref{conj:LDP-fails}. By contrast, the only known algorithmic frameworks that achieve large independent sets in this regime are online. The arguments of~\cite{wein2020optimal,rahman2017local} crucially rely on algorithmic stability, whereas online algorithms are known to be unstable~\cite{gamarnik2025optimal}. Thus, the dense setting calls for substantially different techniques—see Section~\ref{subsection: ogp prior work} for a further discussion.

In a recent development, Gamarnik, K{\i}z{\i}lda\u g, and Warnke~\cite{gamarnik2025optimal} introduced an OGP-based framework tailored for the online setting. This approach is inherently algorithm-dependent: it utilizes a stopping time to monitor output size, together with interpolation paths that evolve temporally as the algorithm progresses. It yielded sharp lower bounds within online algorithms for the classical \ER random graph $G(n,p)$. Extending these techniques to higher uniformities and multipartite models under a global balancedness constraint turns out to be substantially more challenging, due to additional technical hurdles that are absent in $G(n,p)$ (discussed below). Among other technical contributions, we introduce a new stopping-time argument tailored specifically to hypergraphs, which may be of independent interest. Taken together, our techniques allow us to complete the picture for hypergraph models: while the sparse regime was recently addressed by Dhawan and Wang~\cite{dhawan2026low}, the dense regime had remained open, which is our focus.

Formally, we construct a family of correlated random hypergraphs that remain identical until a specific random stopping time, after which they evolve independently. Executing the algorithm on these correlated instances yields a collection of candidate independent sets with rigid structural constraints—most notably, they must coincide on all vertices revealed prior to the stopping time. We define these configurations as ``forbidden tuples'' and prove they are absent from the correlated random hypergraphs with high probability. Since a successful online algorithm would necessarily imply the existence of such tuples, this contradiction establishes our lower bound.

We next extend our analysis to the $H(r, n, p)$ model. Although the statistical threshold is determined via a second-moment argument, the $\gamma$-balancedness, a global constraint, adds substantial technical complexity. As for developing an online algorithm, the primary hurdle in this model is the global nature of the balancedness requirement, which must be maintained throughout the construction. To address this, we introduce a \textit{staged and bucketed} greedy algorithm. The procedure is partitioned into $r$ distinct stages; in each stage, the algorithm populates the independent set using vertices from ``unexhausted'' parts until some such part reaches its target capacity and is ``locked.'' By controlling acceptance probabilities and bounding failure events via careful estimates, we prove that this algorithm yields a $\gamma$-balanced independent set of size $(1-\epsilon)\comppartite$ with high probability.

The derivation of the impossibility result for $H(r,n,p)$ is substantially more involved. The main difficulties stem from two sources: (i) the multipartite nature  of $H(r,n,p)$, together with the permutation invariance of the parts in the partition, and (ii) the $\gamma$-balancedness constraint, which imposes prescribed cardinality requirements on the intersections of the independent set with the $r$ parts. Among other technical ideas, we introduce a novel stopping time that accounts for both issues simultaneously. Specifically, we stop when the independent set contains sufficiently many vertices from all parts of the partition (by contrast, the stopping time for $H_r(n,p)$ is triggered once the algorithm’s output reaches a prescribed size threshold).

\paragraph{Background on OGP} 
As mentioned earlier, our algorithmic lower bound is based on the OGP framework pioneered by Gamarnik and Sudan~\cite{gamarnik2014limits,gamarnik2017} (and termed in~\cite{gamarnik2018finding}). The OGP is among the most powerful techniques for identifying computationally hard phases of random optimization problems—i.e., those in which the objective function is random. Its origins lie in earlier works on random CSPs, where the apparent hard phases were observed to exhibit striking geometric phase transitions~\cite{achlioptas2006solution,achlioptas2008algorithmic,mezard2005clustering}. The OGP framework systematically converts such geometric features of the solution landscape into formal algorithmic lower bounds; see~\cite{gamarnik2021overlap,gamarnik2025optimal} for surveys on the method.

We now describe the OGP in greater detail for random graphs—the context in which it was originally introduced. For the sparse \ER model $G(n,\frac{d}{n})$, the work of Gamarnik and Sudan~\cite{gamarnik2014limits,gamarnik2017} established that independent sets of size greater than $(1+1/\sqrt{2})n\frac{\log d}{d}$ must either substantially overlap or be nearly disjoint.\footnote{For sparse random graphs, the relevant asymptotic regime is the double limit $n\to\infty$ followed by $d\to\infty$.} As a consequence, they rigorously ruled out local algorithms at this threshold, thereby refuting a well-known conjecture by Hatami, Lov\'asz, and Szegedy~\cite{hatami2014limits}, widely believed at the time. A subsequent work by Rahman and Vir\'ag~\cite{rahman2017local} extended this threshold down to $n\frac{\log d}{d}$ by analyzing the overlap patterns of many independent sets, below which polynomial-time algorithms are known~\cite{lauer2007large}. This approach, now known as the multi OGP ($m$-OGP), has since proved instrumental in obtaining sharp algorithmic lower bounds in a variety of other models. In the context of random graphs, hardness at the $(1+1/\sqrt{2})n\frac{\log d}{d}$ threshold was later extended to LDPs~\cite{gamarnik2020low}, and sharp lower bounds for LDPs were subsequently obtained in~\cite{wein2020optimal} using an asymmetric variant of the $m$-OGP. Still more delicate refinements of the $m$-OGP were required in other models, such as random $k$-SAT~\cite{bresler2021algorithmic} and spin glasses~\cite{huang2021tight,huang2023algorithmic,huang2025strong}. For the latter, this led to the branching OGP, an especially powerful variant organized around an ultrametric tree of solutions. {Even more recently, novel variants of the OGP tailored for online algorithms—which need not be stable—have begun to emerge~\cite{gamarnik2025optimal,dhawan2025sharp}; see Section~\ref{subsection: ogp prior work} for details.} {See Section~\ref{sec:prior-scg} for further background on statistical-computational gaps in the context of random graphs and beyond. A summary of our contributions along with potential future directions of inquiry are now in order.}

\subsection{Summary and Open Problems}

In this work, we consider the computational hardness of finding large ($\gamma$-balanced) independent sets in random hypergraphs.
We focus on online algorithms, one of the most natural algorithmic frameworks containing Karp's original algorithm for $G(n, p)$, which motivated his celebrated conjecture regarding statistical-computational gaps for independent sets in $G(n, p)$.
Our results imply a computational gap of a multiplicative factor of $r^{1/(r-1)}$ for independent sets in $H_r(n, p)$ and a gap of a multiplicative factor of $(\max_{i}\gamma_i)^{-1/(r-1)}$ for $\gamma$-balanced independent sets in $H(r, n, p)$ with respect to online algorithms; these gaps match those shown for LDPs in the sparse regime~\cite{dhawan2026low}.

We conclude this introduction with a description of potential future directions of inquiry.

\paragraph{Future Queries}
The setting of~\cite{gamarnik2025optimal} permits querying a limited set of \emph{future edges}—edges incident to vertices not yet seen—at each step. That is, the decision at time $t$ is based not only on the edges revealed so far, but also on a restricted set of such future edges. In this augmented model,~\cite{gamarnik2025optimal} prove both lower bounds and algorithmic guarantees for independent sets in graphs: specifically, they show that algorithms with modest access to future information can in fact exceed the $(1+\epsilon)\log_b n$ threshold, albeit using quasi-polynomial time.
Similarly,~\cite{dhawan2025sharp} prove analogous results for $\gamma$-balanced independent sets in \ER bipartite graphs.

This naturally raises the following question: can algorithms with limited future information outperform the $(1+\epsilon)\compunif$ threshold for $H_r(n,p)$ and the $(1+\epsilon)\comppartite$ threshold for $H(r, n,p)$? We conjecture that the answer is yes:

\begin{conjecture}\label{conj:Future}
Let $E$ be the set of all future edges ever revealed to the algorithm and denote by $\binom{\A(G)}{r}$ the set of all $r$-edges whose vertices are contained in the output $\A(G)$.
\begin{enumerate}
    \item For any $\epsilon > 0$ and $r\in\mathbb{N}$, there exists a constant $c_{\epsilon,r}>0$ and an online algorithm $\mathcal{A}$ that finds an independent set of size at least $(1+\epsilon)\compunif$ whp in $H_r(n, p)$, provided
    \[
    \left|E\cap \binom{\mathcal{A}(G)}{r}\right|\le c_{\epsilon,r} \Bigl(\log_b (np)\Bigr)^{\frac{r}{r-1}}.
    \]

    \item\label{item: partite future} For any $\epsilon > 0$ and $r\in\mathbb{N}$, there exists a constant $c_{\epsilon,r}>0$ and an online algorithm $\mathcal{A}$ that finds a $\gamma$-balanced independent set of size at least $(1+\epsilon)\comppartite$ whp in $H(r, n, p)$, provided
    \[
    \left|E\cap \binom{\mathcal{A}(G)}{r}\right|\le c_{\epsilon,r} \Bigl(\log_b (np)\Bigr)^{\frac{r}{r-1}}.
    \]
\end{enumerate}  
\end{conjecture}
\begin{remark}
Due to $r$-uniformity, the amount of future information required to surpass the computational threshold is $O\left(\log^{\frac{r}{r-1}}n\right)$. For $r=2$, this is the same scaling arising in~\cite{gamarnik2025optimal,dhawan2025sharp}.
\end{remark}

\paragraph{$\beta$-Independent Sets}

For an integer $1 \leq \beta \leq r-1$, a $\beta$-independent set in a hypergraph $H = (V, E)$ is a set $I \subseteq V$ such that $|e\cap I| \leq \beta$ for each $e \in E(H)$.
These structures interpolate between the notions of \textit{strong} and \textit{weak} independent sets—the cases $\beta = 1$ and $\beta = r-1$ (note that $\beta = r-1$ corresponds to the usual notion of independent sets).
While strong and weak independent sets have been heavily studied in the literature (both in the random and deterministic settings~\cite{krivelevich1998chromatic, frieze2013coloring, dhawan2026fractional, ajtai1982extremal, cutler2013hypergraph}), $\beta$-independent sets for $1 < \beta < r-1$ are not as well understood.

Regarding $H_r(n, p)$, the statistical threshold was determined by~\cite{krivelevich1998chromatic} in the sparse regime.\footnote{They actually prove a stronger result in terms of the coloring variant of the parameter.}
A curious feature of these structures is that the online greedy algorithm is not guaranteed to construct a $\beta$-independent set.
To see this, consider an instance $H \sim H_r(n, p)$ and suppose, for example, at step $t$ we have constructed an independent set $I_t$ thus far. At this point, we do not know the structure of the subhypergraph $H' \coloneqq H[I_t \cup \{v_{t+1}, \ldots, v_n\}]$.
In particular, while $I_t$ is a $\beta$-independent set in $H[\{v_1, \ldots, v_t\}]$, it may not be one in $H'$.
Generalizing our results on $H_r(n, p)$—and, indeed, those of~\cite{dhawan2026low} in the sparse regime—to $\beta < r-1$ is an interesting direction.

\paragraph{Distance-$k$ Independent Sets}

In a similar flavor, a distance-$k$ independent set in a graph $G = (V, E)$ is a set $I \subseteq V$ such that the vertices in $I$ are pairwise of distance at least $k + 1$; equivalently, it is an independent set in $G^{k}$ (note that one recovers the usual notion of independent sets for $k = 1$).
As above, the greedy algorithm is not guaranteed to construct a distance-$k$ independent set in a graph $G$.

Very recently, the growth rate of the statistical threshold for the distance-$k$ independence number was determined for $G(n, p)$~\cite{frieze2026coloring}.
They do not pin down the leading constant factor and raise it as an open problem (see the discussion after~\cite[Conjecture~8]{frieze2026coloring}).
In the deterministic setting, there are a number of results for special graph classes, most notably for line graphs~\cite{kaiser2014distance, bi2026strong, mahdian2000strong}.
We remark that the problem is trivial for random graphs when $p \geq \sqrt{\frac{2\log n}{n}}$ and $k \geq 2$ as such graphs have diameter $2$~\cite{bollobas1981diameter}.
It would be interesting to investigate the problem in the sparse regime.

\begin{question}
    Are there statistical--computational gaps for the maximum distance-$k$ independent set problem when $p = o\left(\sqrt{\frac{2\log n}{n}}\right)$ and $k \geq 2$?
\end{question}

\section{Main Results}\label{sec: main results}
In this section, we state our main results formally.
We first define the random models we consider.

\begin{definition}[The Random Hypergraph Models]\label{definition: models}
    Let $n,\, r \in \N$ such that $n \geq r \geq 2$.
    \begin{itemize}
        \item We construct the hypergraph $H \sim H_r(n, p)$ on vertex set $[n]$ by including each $e \subseteq [n]$ of size $r$ in $E(H)$ independently with probability $p$.
        \item We construct the hypergraph $H \sim H(r, n, p)$ on vertex set $[n]\times [r]$ by including each
        \[e \in V_1 \times \cdots \times V_r,\]
        in $E(H)$ independently with probability $p$.
        Here, $V_i = [n] \times \set{i}$.
    \end{itemize}
\end{definition}

As mentioned earlier, our goal is to understand the limitations of online algorithms in constructing large independent sets in $H_r(n, p)$ and large $\gamma$-balanced independent sets in $H(r, n, p)$.
Before formally describing the algorithmic framework, we define  $\gamma$-balanced independent sets. Ash first introduced these structures~\cite{ash1983two} in the special case $r = 2$ and $\gamma = (1/2, 1/2)$, where he aimed to determine when a bipartite graph contains a Hamilton cycle; this setting has been extensively studied from a combinatorial standpoint giving rise to the field of bipartite Ramsey theory (see, e.g., ~\cite{chakraborti2023extremal, feige1992hardness, axenovich2021bipartite, bal2020bipartite}).
The general setting $\gamma = (\beta, 1 - \beta)$ for $r = 2$ was introduced by Perkins and Wang~\cite{perkins2024hardness}.
Dhawan generalized Ash's definition to $r\geq 3$ for $\gamma = (1/r, \ldots, 1/r)$~\cite{dhawan2025balanced}; in subsequent work with Wang~\cite{dhawan2026low}, they introduced the definition for arbitrary $\gamma \in \Q_+^r$ which we now state:

\begin{definition}[$\gamma$-Balanced independent sets]\label{definition: balanced independent sets}
    Let $H = (V_1\sqcup \cdots \sqcup V_r,\, E)$ be an $r$-uniform $r$-partite hypergraph for $r \geq 2$, and let $\gamma \in \Q_+^r$ be such that $\sum_{i = 1}^r\gamma_i = 1$.
    An independent set $I \subseteq V(H)$ is \textit{$\gamma$-balanced} if $|I\cap V_{\sigma(i)}| = \gamma_i|I|$ for some permutation $\sigma \,:\, [r] \to [r]$ and each $i \in [r]$.
\end{definition}

With the above definitions in hand, we are ready to describe the framework of \textit{online algorithms} with respect to each of these models.

\begin{definition}[Online arrival model]\label{def : online-arrival}
Let $H = (V,E)$ be a random $r$-uniform hypergraph drawn from $H_r(n,p)$ (resp.~$H(r,n,p)$).
A randomized algorithm $\mathcal{A}$ with internal randomness determined by seed $\omega$ runs for $n$ rounds (resp.~$rn$ rounds) and keeps track of the set $S_t$ of vertices revealed thus far; initially $S_t = \emptyset$. 
At each round $t$:
\begin{enumerate}
\item Based on $\omega$ and all information revealed thus far, $\mathcal{A}$ randomly selects a vertex $v_t\in V\setminus S_{t-1}$ and reveals the status of all hyperedges incident to $v_t$ consisting of vertices in $S_{t-1} \cup \{v_t\}$.
That is, for every $T \subseteq S_{t-1}$ with $|T| = r-1$, the algorithm learns whether $T \cup \{v_t\} \in E(H)$ whenever such a tuple is admissible (i.e., always admissible in $H_r(n,p)$ and restricted to cross-part tuples in $H(r,n,p)$).

\item Based on $\omega$ and all information revealed thus far, $\mathcal{A}$ then irrevocably decides if $\mathcal{A}_t(G) = \mathcal{A}_{t-1}(G)\cup \{v_t\}$.
The final output is a set $I \subseteq V$, which is required to be an independent set in $H_r(n, p)$ and a $\gamma$-balanced independent set in $H(r, n, p)$.
\end{enumerate}
\end{definition}

Per Definition~\ref{def : online-arrival}, the vertex arrival order is random, determined jointly by the algorithm’s internal randomness $\omega$ and the randomness of $H$ (its edges). 
When considering $H(r, n, p)$, there is a subtlety to note regarding the information revealed at each step.
In particular, since edges must go across all parts, it is possible that no new information is revealed for multiple steps—for example, it is possible that the algorithm selects vertices from $\cup_{i < r}V_i$ for the first $(r-1)n$ rounds.

We remark that our results hold for the most general setting: (i) the algorithmic bounds (Theorems~\ref{theo: a_COMP-achievable-Hr} and~\ref{theo: a_comp-achievable-Hrnp}) are independent of the arrival order, and (ii) the hardness results (Theorem~\ref{theo: a_COMP-impossible-Hr} and~\ref{theo: a_COMP-impossible-Hrnp}) apply to all online arrival scenarios allowed by Definition~\ref{def : online-arrival}.
Our focus is on online algorithms that return large independent sets with specified probability, formalized as follows.

\begin{definition} \label{def: kdelta-Hr}
    For parameters $k > 0$ and $\delta \in [0,1]$, an online algorithm $\mathcal{A}$ operating according to Definition~\ref{def : online-arrival} is said to $(k, \delta)$-optimize the independent set problem in $H_r(n, p)$ (resp.~the $\gamma$-balanced independent set problem in $H(r, n, p)$) if the following is satisfied when $H \sim H_r(n, p)$ (resp.~$H\sim H(r, n, p)$):
    \[\P[|\mathcal{A}(G)| \geq k] \geq \delta.\]
\end{definition}

We are now ready to state our main results.
We consider each model separately.

\subsection{Independent sets in $H_r(n, p)$}
Recall the statistical and computation thresholds for $H_r(n, p)$ from Section~\ref{subsection: informal results}:
\begin{equation}\label{eq:Thresholds unif}
    \statunif \coloneqq \left(r!\log_b n\right)^{\frac{1}{r-1}}\qquad\text{and}\qquad \compunif\coloneqq \left((r-1)!\log_b n\right)^{\frac{1}{r-1}},
\end{equation}
where $b = 1/(1-p)$.
We begin by determining the size of the largest independent set in $H_r(n,p)$ for constant $p$.
\begin{theorem}\label{theo: a_stat-Hr}
Let $\epsilon > 0$, $p \in (0, 1)$, and $r \in \N$ such that $r\geq 2$.
The following holds for $n$ sufficiently large.
Let $H \sim H_r(n,p)$ and let $Z_\alpha$ denote the
number of independent sets of size $\alpha$ in $H$.
For $\statunif$ as defined in~\eqref{eq:Thresholds unif}, the following hold:
 \begin{enumerate}[label=(H\arabic*)]
    \item\label{item: z fmm} For $\alpha \ge (1+\epsilon)\,\statunif$, $\P[Z_\alpha>0] \xrightarrow{n\to\infty}{} 0$.
    \item\label{item: z smm} For $\alpha \le (1-\epsilon)\,\statunif$, $\P[Z_\alpha>0] \xrightarrow{n\to\infty}{} 1$.
\end{enumerate}
\end{theorem}

Thus, the largest independent set is approximately of size $\statunif$, which we refer to as the \emph{statistical threshold}. 
The proof of Theorem~\ref{theo: a_stat-Hr} follows a standard application of the first and second moment method; see Section~\ref{sec: stat-threshold-unif} for the details.

We now present our achievability result, i.e., we show that there is an online algorithm that constructs a large independent in $H_r(n, p)$ with high probability.

\begin{theorem} \label{theo: a_COMP-achievable-Hr}
Let $\epsilon > 0$, $p \in (0, 1)$, and $r \in \N$ such that $r\geq 2$.
The following holds for $n$ sufficiently large.
There is an online algorithm $\mathcal{A}$ that $(k, \delta)$-optimizes the independent set problem in $H_r(n, p)$, where
\[k = (1-\epsilon)\compunif, \qquad \text{and} \qquad \delta = 1 - \exp\left(-n^{\Theta(\epsilon)}\right).\]
\end{theorem}

See Section~\ref{sec: achieve-unif} for the proof.
We next complement Theorem~\ref{theo: a_COMP-achievable-Hr} with a sharp lower bound.

\begin{theorem}\label{theo: a_COMP-impossible-Hr}
Let $\epsilon > 0$, $p \in (0, 1)$, and $r \in \N$ such that $r\geq 2$.
The following holds for $n$ sufficiently large.
There exists no online algorithm that $(k, \delta)$-optimizes the independent set problem in $H_r(n, p)$, where
\[k = (1+\epsilon)\compunif, \qquad \text{and} \qquad \delta = \exp\left(-O\left(\epsilon^{2} \log_b ^{r/(r-1)} n\right)\right).\]
\end{theorem}
Consequently, online algorithms exhibit a
statistical–computational gap of a multiplicative factor of $r^{1/(r-1)}$,
matching the gap shown for LDPs in the sparse regime~\cite{dhawan2026low}.

\subsection{Balanced Independent sets in $H(r, n, p)$}

Recall the statistical and computation thresholds for $H(r, n, p)$ from Section~\ref{subsection: informal results}:
\begin{equation}\label{eq:Thresholds partite} 
    \statpartite\coloneqq \left(\frac{\log_b n}{\prod_{i=1}^r \gamma_i}\right)^{\frac{1}{r-1}}\qquad\text{and}\qquad \comppartite\coloneqq \left(\frac{\log_b n}{\prod_{i=1}^r \gamma_i}\max_{1\le i\le r}\gamma_i\right)^{\frac{1}{r-1}},
\end{equation}
where $b = 1/(1-p)$.
We begin by determining the size of the largest $\gamma$-balanced independent set in $H(r, n,p)$ for constant $p$.

\begin{theorem} \label{theo: gamma_bal a_stat}
Let $\epsilon > 0$, $p \in (0, 1)$, $r \in \N$ such that $r\geq 2$, and let $\gamma \in \Q_+^r$ be such that $\sum_i\gamma_i = 1$.
The following holds for $n$ sufficiently large.
Let $H \sim H(r, n,p)$ and let $Z_\alpha(\gamma)$ denote the number of $\gamma$-balanced independent sets of size $\alpha$ in $H$.
For $\statpartite$ as defined in \eqref{eq:Thresholds partite}, the following hold:
 \begin{enumerate}[label=(M\arabic*)]
    \item\label{item: z fmm partite} For $\alpha \ge (1+\epsilon)\,\statpartite$, $\P[Z_\alpha(\gamma)>0] \xrightarrow{n\to\infty}{} 0$.
    \item\label{item: z smm partite} For $\alpha \le (1-\epsilon)\,\statpartite$, $\P[Z_\alpha(\gamma)>0] \xrightarrow{n\to\infty}{} 1$.
    \end{enumerate}
\end{theorem}

Thus, the largest $\gamma$-balanced independent set is approximately of size $\statpartite$, which we refer to as the \emph{statistical threshold}. 
The proof of Theorem~\ref{theo: gamma_bal a_stat} follows a standard application of the first and second moment method; see Section~\ref{sec: stat-threshold-rpart} for the details.

We now present our achievability result, i.e., we show that there is an online algorithm that constructs a large $\gamma$-balanced independent in $H(r, n, p)$ with high probability.

\begin{theorem}\label{theo: a_comp-achievable-Hrnp}
Let $\epsilon > 0$, $p \in (0, 1)$, and $r \in \N$ such that $r\geq 2$, and let $\gamma \in \Q_+^r$ be such that $\sum_i\gamma_i = 1$.
The following holds for $n$ sufficiently large.
There is an online algorithm $\mathcal{A}$ that $(k, \delta)$-optimizes the $\gamma$-balanced independent set problem in $H(r, n, p)$, where
\[k = (1-\epsilon)\comppartite, \qquad \text{and} \qquad \delta = 1 - \exp\left(-\Omega\left(n^{\epsilon/(r-1)}\right)\right).\]
\end{theorem}

See Section~\ref{sec: achiv-rpart} for the proof.
We next complement Theorem~\ref{theo: a_comp-achievable-Hrnp} with a sharp lower bound.

\begin{theorem}\label{theo: a_COMP-impossible-Hrnp}
Let $\epsilon > 0$, $p \in (0, 1)$, and $r \in \N$ such that $r\geq 2$, and let $\gamma \in \Q_+^r$ be such that $\sum_i\gamma_i = 1$.
The following holds for $n$ sufficiently large.
There exists no online algorithm that $(k, \delta)$-optimizes the $\gamma$-balanced independent set problem in $H(r, n, p)$, where
\[k = (1+\epsilon)\comppartite, \qquad \text{and} \qquad \delta = \exp\left(-O\left(\epsilon^{2} \log_b ^{r/(r-1)} n\right)\right).\]

\end{theorem}

The proof of this result can be found in Section~\ref{subsec:imposs_Hrnp}. Consequently, online algorithms exhibit a
statistical–computational gap of a multiplicative factor of $(\max_{i}\gamma_i)^{-1/(r-1)}$,
matching the gap shown for LDPs in the sparse regime~\cite{dhawan2026low}.
Note that this gap matches that for independent sets in $H_r(n, p)$ when considering $\gamma = \mathbf{1}/r$ (here, $\mathbf{1}$ denotes the all ones vector in $\Q^r$), i.e., the truly balanced case where the independent set contains an equal number of vertices from each part.

\section{Further Background}\label{sec:prior-scg}

Statistical-computational gaps - i.e., gaps between what is information-theoretically possible and what is achievable by known polynomial-time algorithms - are a central feature in many average-case models. Examples include optimization problems over random graphs (which is also our main focus)~\cite{gamarnik2014limits,gamarnik2017,gamarnik2020low,wein2020optimal,perkins2024hardness,dhawan2026low,gamarnik2025optimal,dhawan2025sharp},  random constraint satisfaction problems (CSPs)~\cite{achlioptas2006random,achlioptas2006solution,achlioptas2008algorithmic,gamarnik2017performance,bresler2021algorithmic,kizildaug2023sharp,yung2024,mezard2005clustering}, spin glasses~\cite{chen2019suboptimality,gamarnikjagannath2021overlap,gamarnik2025shattering,huang2021tight,huang2023algorithmic,huang2025strong,sellke2025tight, sen2018optimization}, as well as other models such as number balancing~\cite{gamarnik2023algorithmic,mallarapu2025strong,kizildaug2023planted}, discrepancy minimization~\cite{gamarnik2022algorithms,gamarnik2023geometric,li2024discrepancy}, graph alignment~\cite{du2025algorithmic}, largest subtensor problem~\cite{bhamidi2025finding,kr2026large} and beyond. Due to randomness, classical complexity theory often offers little insight into such problems (see, e.g.,~\cite{ajtai1996generating,boix2021average,gamarnik2018computing,vafa2025symmetric} for several notable exceptions). As a result, a large body of work has developed alternative frameworks to probe the `apparently hard’ phase of these average-case models. We do not discuss these frameworks in detail here and instead refer the reader to the excellent surveys~\cite{bandeira2018notes,kunisky2019notes,gamarnik2021overlap,gamarnik2025turing,gamarnik2022disordered,wein2025computational,wu2018statistical} for broader background.

\subsection{Optimization and Inference on Random Hypergraphs}

While random graph inference has been a cornerstone of research for decades, its extension to hypergraphs remains relatively uncharted territory. 
The hypergraph setting is generally regarded as significantly more complex. 
Nevertheless, recent years have seen a surge in both theoretical and practical interest in the field. 
This momentum is evidenced by the use of spectral methods for testing and estimation problems~\cite{luo2022tensor, jones2023sum}, as well as ongoing efforts to identify statistical and computational thresholds for planted variants~\cite{yuan2021heterogeneous, yuan2021information, dhawan2023detection}.

Although the maximum independent set problem is NP-hard, significant advancements have been made by focusing on restricted algorithmic classes within structurally constrained hypergraphs. Specifically, research has addressed bounded-degree hypergraphs~\cite{halldorsson2009independent, guruswami2011complexity}, streaming algorithms for sparse hypergraphs~\cite{halldorsson2016streaming}, semi-random models~\cite{khanna2021independent}, and semidefinite programming (SDP) approaches~\cite{halperin2002improved, agnarsson2013sdp}. Notably,~\cite{dhawan2026low} provided the first investigation into the statistical–computational gap for finding independent sets in \ER hypergraphs, specifically targeting LDPs in the sparse regime.

In contrast, multipartite hypergraphs remain relatively under-explored from a computational perspective, though several theoretical works offer valuable foundations~\cite{kamvcev2017bounded, dhawan2023list, bowtell2024matchings}. One notable study by~\cite{guruswami2015inapproximability} investigates the minimum vertex cover problem in these structures, revealing a striking complexity shift: the problem is tractable for $r = 2$ but becomes NP-hard for $r \geq 3$. A parallel phenomenon was observed by~\cite{botelho2012cores} regarding the emergence of a $k$-core (a subhypergraph with a minimum degree at least $k$) within certain random multipartite hypergraph models.
Interestingly,~\cite{dhawan2026low}—the first to analyze the statistical–computational gap for independent sets in Erdős–Rényi hypergraphs—noted that the maximum balanced independent set problem does not exhibit this same sensitivity to uniformity. This divergence in behavior highlights the unique complexity landscape of multipartite hypergraphs compared to their standard counterparts.

\subsection{OGP for Online Algorithms}\label{subsection: ogp prior work}
As discussed previously, at a high level, the OGP 
asserts that solutions of certain optimization problems at intermediate distances do not exist—they either overlap substantially, or are basically disjoint. The method is primarily tailored for algorithms that exhibit input stability and, by virtue of smooth evolution under perturbations of the input, cannot overcome OGP-based barriers.\footnote{Indeed, there are models such as the shortest path which exhibit OGP, yet remain amenable to linear programing; such algorithms are known to be unstable~\cite{li2024some}.} Stable algorithms form a broad class and include many prominent paradigms such as local algorithms~\cite{gamarnik2017,gamarnik2017performance,rahman2017local}, approximate message passing~\cite{gamarnikjagannath2021overlap}, low-degree polynomials~\cite{gamarnik2020low,wein2025computational}, low-depth Boolean circuits~\cite{gamarnik2021circuit}, gradient descent and Langevin dynamics~\cite{gamarnik2020low}, as well as other miscellaneous algorithms~\cite{gamarnik2022algorithms}. 
\begin{remark}
For models such as number balancing and binary perceptrons, the OGP-based hardness guarantees were supported by classical evidence of hardness (e.g., based on lattices) in certain parameter regimes~\cite{vafa2025symmetric}. 
\end{remark}

For dense random graphs, however, the only known successful polynomial-time algorithms are online, which may well be unstable even in the vanilla \ER model $G(n,p)$~\cite[Proposition~1.1]{gamarnik2025optimal}. The first OGP-based argument for online algorithms was established in the context of binary perceptron~\cite{gamarnik2023geometric} via a resampling argument, and was later extended to graph alignment~\cite{du2025algorithmic} and largest subtensor problems~\cite{bhamidi2025finding}. These works combine existing $m$-OGP barriers along with a careful resampling scheme to derive lower bounds.

Moving beyond is substantially more challenging, even for $G(n,p)$. Intuitively, the main obstacle is precisely the absence of stability—the very notion classical OGP-based barriers build upon. For $G(n,p)$ with constant $p$, the first sharp lower bounds for online algorithms were obtained by Gamarnik, K{\i}z{\i}lda\u{g}, and Warnke~\cite{gamarnik2025optimal} using novel technical refinements. While classical OGP-based barriers are largely oblivious to the algorithm itself, the argument in~\cite{gamarnik2025optimal} develops judiciously designed temporal interpolation paths that evolve with the algorithm and account for the online structure. Additionally, their argument uses a stopping time to track the input size, and constructs the interpolation paths randomly. These ingredients represent substantial departures from the classical OGP toolkit. Subsequent work of Dhawan, K{\i}z{\i}lda\u{g}, and Maitra~\cite{dhawan2025sharp} extended these to the balanced independent set problem in random bipartite graphs, where incorporating the balancedness constraint and handling vertex arrival orders (which is random and may potentially reveal relatively few cross-edges and limited information) introduce delicate technical complications.

One of the main motivations of the present paper is to pin down the algorithmic thresholds for dense hypergraph models. Given that the sparse setting is by now relatively well understood, our broader goal is to develop a comparably robust theory for dense random graphs and hypergraphs.

We close this section by noting that online algorithms play a central role in machine learning and optimization, especially in the modern era of `big data and AI'~\cite{hazan2016introduction,rakhlin2010online,rakhlin2011online-a,rakhlin2011online-b,rakhlin2013online}. Many practically relevant algorithms are online while failing to satisfy the kinds of stability assumptions underlying classical OGP arguments, which further underscores the need for a more complete toolkit for online models.

\section{Independent Sets in $H_r(n, p)$} 

In this section, we prove our results pertaining to $H_r(n, p)$.
We further split the section into three subsections devoted to the proofs of Theorems~\ref{theo: a_stat-Hr},~\ref{theo: a_COMP-achievable-Hr}, and~\ref{theo: a_COMP-impossible-Hr}, respectively.

\subsection{Statistical Threshold}\label{sec: stat-threshold-unif}
In this section, we prove Theorem~\ref{theo: a_stat-Hr}. A fixed $\alpha$-set is independent if and only if none of its $\binom{\alpha}{r}$ $r$-tuples is present in the hypergraph, which occurs with probability $(1-p)^{\binom{\alpha}{r}}=b^{-\binom{\alpha}{r}}$. Hence, for every $\alpha$, we have
\begin{equation}\label{eq: expected Z}
    \mathbb{E}[Z_{\alpha}] = \binom{n}{\alpha} b^{-\binom{\alpha}{r}}.
\end{equation}

Let $\alpha \geq (1 + \epsilon) \statunif$, where $\statunif$ is defined in \eqref{eq:Thresholds unif}.
We have
\begin{align*}
    \mathbb{E}[Z_{\alpha}] & \leq n ^ {\alpha} b ^ {- \binom{\alpha}{r}} = \exp\left(\alpha \log n - \binom{\alpha}{r} \log b\right)  \\
    & \leq \exp\left(\alpha \log n - (1 - o(1))\frac{\alpha^r}{r!} \log b\right)\ \\
    & = \exp\left(\alpha \log n\left(1 - (1 - o(1))\frac{\alpha^{r-1}}{r!\log_bn}\right)\right)\ \\
   &\le \exp\left(\alpha \log n\left(1 - (1 - o(1))(1+\epsilon)^{r-1}\right)\right)\ \\
    &= \exp\left(-\Omega\left((\log_b n)^{\frac{r}{r-1}}\right)\right)
    \xrightarrow{n\to\infty}{}0.
\end{align*}
Thus,~\ref{item: z fmm} follows by Markov's inequality.

For $\alpha \le \alpha '$, we have $Z_\alpha \ge Z_{\alpha '}$. Therefore, to prove~\ref{item: z smm}, it suffices to show the claim holds for $\alpha = \alpha_\epsilon \coloneqq (1 - \epsilon) \statunif$. We define
\[
Z_{\alpha_\epsilon} = \sum_{\substack{S\subseteq [n]\\ |S|=\alpha_\epsilon}} I_S,
\qquad
\text{where}
\qquad 
I_S \coloneqq \mathbf{1}\{S \text{ is an independent set in } H_r(n,p)\}.
\]
Then
\begin{equation}\label{eq: z squared}
    Z_{\alpha_\epsilon}^2 =\sum_{\substack{S\subseteq [n]\\ |S|=\alpha_\epsilon}} \ \sum_{\substack{T\subseteq [n]\\ |T|=\alpha_\epsilon}} I_S I_T. 
\end{equation}
We parameterize ordered pairs $(S,T)$ by $m \coloneqq |S\cap T|.$ Clearly, $0\le m \le \alpha_\epsilon$.
Fix $S \subseteq V(H)$ with $|S|=\alpha_\epsilon$. The number of $T \subseteq V(H)$ such that $|T|=\alpha_\epsilon$ and $|S\cap T|=m$ is 
\begin{equation}\label{eq: n(m)}
    N(m) = \binom{\alpha_\epsilon}{m}\binom{n-\alpha_\epsilon}{\alpha_\epsilon-m}. 
\end{equation} 
Note that
\begin{equation}\label{eq: expected is,t}
    \mathbb{E}[I_S I_T] = (1-p)^{2\binom{\alpha_\epsilon}{r}-\binom{|S\cap T|}{r}}=b^{-\left(2\binom{\alpha_\epsilon}{r}-\binom{|S\cap T|}{r}\right)}. 
\end{equation}
Summing over $m$ and combining \eqref{eq: z squared}, \eqref{eq: n(m)}, and \eqref{eq: expected is,t} yields 
\[
\mathbb{E}[Z_{\alpha_\epsilon}^2]
= \binom{n}{\alpha_\epsilon}\sum_{m=0}^{\alpha_\epsilon}
N(m)
\,b^{-\left(2\binom{\alpha_\epsilon}{r}-\binom{m}{r}\right)}.
\]
By \eqref{eq: expected Z}, we have $\mathbb{E}[Z_{\alpha_\epsilon}]
=\binom{n}{\alpha_\epsilon}b^{-\binom{\alpha_\epsilon}{r}}$, which implies
\[
\frac{\mathbb{E}[Z_{\alpha_\epsilon}^2]}{\mathbb{E}[Z_{\alpha_\epsilon}]^2}
= 1+\sum_{m=1}^{\alpha_\epsilon}
\frac{\binom{\alpha_\epsilon}{m}\binom{n-\alpha_\epsilon}{\alpha_\epsilon-m}}
{\binom{n}{\alpha_\epsilon}}\; b^{\binom{m}{r}} .
\]
Note that
\[
  \frac{\binom{n - \alpha_\epsilon}{\alpha_\epsilon - m}}{\binom{n}{\alpha_\epsilon}} \leq  \frac{\binom{n - m}{\alpha_\epsilon - m}}{\binom{n}{\alpha_\epsilon}} = \frac{(n - m)! (n - \alpha_\epsilon)! \alpha_\epsilon!}{n!(\alpha_\epsilon - m)!(n - \alpha_\epsilon)!} \leq \left(\frac{\alpha_\epsilon}{n}\right)^m.
\]
Applying the above in conjunction with the standard binomial coefficient upper bound, we simplify our earlier expression to
\begin{align*}
    \frac{\mathbb{E}[Z_{\alpha_\epsilon}^2]}{\mathbb{E}[Z_{\alpha_\epsilon}]^2} 
    & \leq 1 + \sum^{\alpha_\epsilon}_{m = 1} \alpha_\epsilon^{m} \left(\frac{\alpha_\epsilon}{n}\right) ^ m b ^{\frac{m^r}{r!}} \\
    & = 1+ \sum^{\alpha_\epsilon}_{m =1} \exp\left(2m \log \alpha_\epsilon - m \log n + \frac{m^r}{r!} \log b \right) \\
    & = 1 + \sum^{\alpha_\epsilon}_{m =1} \exp\left(m\left(2 \log \alpha_\epsilon - \log n + \frac{m^{r-1}}{r!} \log b\right)\right) .
\end{align*}
By the choice of $\alpha_\epsilon$, we have
\[
\frac{m^{r-1}}{r!}\log b \;\le\; \frac{\alpha_\epsilon^{\,r-1}}{r!}\log b
\;\le\; (1-\epsilon)\log n.
\]
Hence, for each $m\ge1$,
\[
m\left(2\log\alpha_\epsilon-\log n+\frac{m^{r-1}}{r!}\log b\right)
\;\le\; m\left(2\log\alpha_\epsilon-\epsilon\log n\right)
= -\,\Omega(\log n),
\]
since $\alpha_\epsilon=(\log n)^{O(1)}$.
It follows that
\[
\sum_{m=1}^{\alpha_\epsilon}\exp \left(-\Omega(\log n)\right)
= \exp \left(-\Omega(\log n)\right)\to 0.
\]
Therefore,
\[
1 \leq \frac{\mathbb{E}[Z_{\alpha_\epsilon}^2]}{\mathbb{E}[Z_{\alpha_\epsilon}]^2}
\le 1+\exp \left(-\Omega(\log n)\right).
\]
Using the Paley--Zygmund inequality, we have
\[
    \P[Z_{\alpha_\epsilon} > 0] \geq \frac{\mathbb{E}[Z_{\alpha_\epsilon}]^2}{\mathbb{E}[Z_{\alpha_\epsilon}]^2} = 1 - \exp \left(-\Omega(\log n)\right)\xrightarrow{n \to \infty}1,
\]
completing the proof.

\begin{remark}[Diverging $r$]\label{rem:div_r_hr}
    Theorem~\ref{theo: a_stat-Hr} and its proof in fact stays true if we allow $r=r_n$ to grow with $n$, as long as (1) we can write $\binom{{\alpha_\epsilon}}{r}=(1+o(1))\frac{\alpha_\epsilon^r}{r!}$, which is true whenever $r^2=o(\alpha_\epsilon)$, and (2) when $\statunif \gg 1$, which is true whenever $r_n = O\left(\frac{\log \log n}{\log \log \log n}\right)$.
\end{remark}

\subsection{Achievability Result} \label{sec: achieve-unif}

In this section, we will prove Theorem~\ref{theo: a_COMP-achievable-Hr}. Our goal is to analyze the size of the independent set produced by an online greedy algorithm. The algorithm proceeds in rounds $t = 1, \dots, n$, where at each round a vertex $v_t$ is revealed according to an online vertex arrival model.
We let $S_t$ denote the set of vertices queried up until and including time $t$, and let $I_t \subseteq S_t$ denote the current independent set, with $I_0 = \emptyset$. Upon revealing $v_t$, the algorithm adds $v_t$ to $I_{t-1}$ if and only if adding it preserves independence, i.e., $I_{t-1} \cup \{v_t\}$ does not contain any hyperedge. 
If $v_t$ is accepted, we set $I_{t} = I_{t-1} \cup \{v_t\}$; otherwise, $I_t = I_{t-1}$. The final output is $I_n$. 
Let us now describe our algorithm formally.

\begin{algorithm}[H]
\caption{Greedy Algorithm for Building an Independent Set in $H_r(n,p)$}
\label{alg:greedy-IS-Hrnp}
$I_0 \gets \emptyset$, \qquad $S_0 \gets \emptyset$,\qquad $t_0 \gets 0$,\qquad $i \gets 0$\ \\
\For{$t \gets 1$ \KwTo $n$}{
  Sample a random vertex $v_t \notin S_{t-1}$ \\
  $S_{t} = \{v_t\} \cup S_{t-1}$ \\
  \If{$v_t \notin I_{t-1}$ \textbf{and} $I_{t-1} \cup \{v_t\}$ is independent}{
    $I_{t} \gets I_{t-1} \cup \{v_t\}$  \\
    $i \gets i + 1$ \\
    $\Delta_i \gets t - t_{0}$ \\
    $t_{0} \gets t$    
  }
  \Else{
    $I_t \gets I_{t-1}$
  }
}
\Return $I_{n}$
\end{algorithm}

For each. $i\ge 1$, define
\[
t_i\coloneqq\min\{n+1,\, \min\{t\in[n]: |I_t|=i\}\},
\qquad
\Delta_i\coloneqq t_i-t_{i-1}.
\]
Thus, if the greedy algorithm never reaches size $i$ by time $n$, then $t_i=n+1$. Note that $\Delta_1,\dots,\Delta_{r-1}$ are deterministically equal
to one.
When the current greedy independent set has size $i-1$, a newly sampled vertex is
accepted if it does \emph{not} lie in any edge together with $r-1$ vertices from the
current independent set. There are
$\binom{i-1}{r-1}$ such possible edges, each present independently with probability $p$.
Hence for $b = 1/(1-p)$, we have
\[
    \P\left[v_t \text{ is added to } I_{t-1}\right]
  = (1-p)^{\binom{i-1}{r-1}} = b^{-\binom{i-1}{r-1}},
\]
for $t_{i-1} < t \leq t_i$. Therefore, conditioned on the history $\mathcal F_{t_{i-1}}$, the waiting time for
the greedy independent set to grow from size $i-1$ to size $i$, truncated by the
remaining time horizon, satisfies
\begin{align}
    \Delta_i \mid \mathcal {F}_{t_{i-1}}
\overset{d}{=}
\min\{\Delta_i',\, n+1-t_{i-1}\},
\qquad
\Delta_i'\sim \mathrm{Geom}(b^{-\binom{i-1}{r-1}}).
\label{eq: delta i}
\end{align}
Recall that $I_t$ denotes the greedy independent set after $t$ sampling steps. We are interested in $|I_n|$, the size after $n$ steps. Observe that
\[
|I_n|\ge k
\quad\Longleftrightarrow\quad
\Delta_1+\cdots+\Delta_k\le n, 
\]
and,
\[
\bigcap_{i=1}^k\{\Delta_i\le n/k\}
\subseteq
\{\Delta_1+\cdots+\Delta_k\le n\},
\]
we have
\begin{align}
\mathbb P[|I_n|\ge k]
&\ge
\mathbb P\left[
\bigcap_{i=1}^k\{\Delta_i\le n/k\}
\right]\notag \\ 
&=
\prod_{i=1}^k
\mathbb P\left[
\Delta_i\le n/k
\,\middle|\,
\Delta_1\le n/k,\ldots,\Delta_{i-1}\le n/k.
\right].
\label{eq: P In-k}
\end{align}
Fix $i$. Since
$\Delta_1,\ldots,\Delta_{i-1}$ have already been revealed by time
$t_{i-1}$, the event \(\{\Delta_1\le n/k,\ldots,\Delta_{i-1}\le n/k\}\) is determined by the history $\mathcal F_{t_{i-1}}$. Hence, by the tower
property,
\[
\begin{aligned}
&\mathbb P\left[
\Delta_i\le n/k
\,\middle|\,
\Delta_1\le n/k,\ldots,\Delta_{i-1}\le n/k
\right] \\
&\qquad =
\mathbb E\left[
\mathbb P\left[
\Delta_i\le n/k
\,\middle|\,
\mathcal F_{t_{i-1}}
\right]
\,\middle|\,
\Delta_1\le n/k,\ldots,\Delta_{i-1}\le n/k
\right].
\end{aligned}
\]
From the conditional distribution~\eqref{eq: delta i}, we have \(\Delta_i \mid \mathcal F_{t_{i-1}} \le \Delta_i'\) and
\[
\mathbb P\left[
\Delta_i>n/k
\,\middle|\,
\mathcal F_{t_{i-1}}
\right]
\le
\mathbb P[\Delta_i'>n/k]
=
\left(1-b^{-\binom{i-1}{r-1}}\right)^{n/k}.
\]
Equivalently,
\[
\mathbb P\left[
\Delta_i\le n/k
\,\middle|\,
\mathcal F_{t_{i-1}}
\right]
\ge
1-
\left(1-b^{-\binom{i-1}{r-1}}\right)^{n/k}.
\]
Since this lower bound holds after conditioning on the full history
$\mathcal F_{t_{i-1}}$, it also holds after averaging over all histories satisfying \(\Delta_1\le n/k,\ldots,\Delta_{i-1}\le n/k.\) Therefore,
\[
\mathbb P\left[
\Delta_i\le n/k
\,\middle|\,
\Delta_1\le n/k,\ldots,\Delta_{i-1}\le n/k
\right]
\ge
1-
\left(1-b^{-\binom{i-1}{r-1}}\right)^{n/k}.
\]
Substituting this bound into~\eqref{eq: P In-k} gives
\[
\mathbb P[|I_n|\ge k]
\ge
\prod_{i=1}^k
\left(
1-
\left(1-b^{-\binom{i-1}{r-1}}\right)^{n/k}
\right)
\ge
\left(
1-
\left(1-b^{-\binom{k-1}{r-1}}\right)^{n/k}
\right)^k.
\]
Applying Bernoulli’s inequality, we further simplify to
\[
  \P\left[|I_n|\ge k\right]
  \ge 1 - k\left(1-b^{-\binom{k-1}{r-1}}\right)^{n/k}.
\]
Using $(1-x) \le e^{-x}$, we obtain
\begin{align}
  \left(1-b^{-\binom{k-1}{r-1}}\right)^{n/k}
  &\leq \exp\left(
      -\frac{n}{k}b^{-\binom{k-1}{r-1}} \right)
\label{eq:ineq}
\end{align}
For $k \coloneqq (1-\epsilon)\compunif$, using $\binom{k-1}{r-1} \leq \frac{k^{r-1}}{(r-1)!}$, we have
\[
b^{-\binom{k-1}{r-1}}
\geq b^{-(1-\epsilon)^{r-1}\log_b n}
= n^{-(1-\epsilon)^{r-1}}.
\]
Plugging this into \eqref{eq:ineq} yields
\[
\left(1-b^{-\binom{k-1}{r-1}}\right)^{n/k}
\leq \exp\left(-\frac{n^{\,1-(1-\epsilon)^{r-1}}}{k}\right)
\le \exp\left(- n^{\Theta(\eps)}\right),
\]
where we use the fact that $k=\mathrm{polylog}(n)$ and
$1-(1-\epsilon)^{r-1}> \epsilon/2$.
In particular,
\[
\mathbb P\left(|I_n|\ge (1-\epsilon)\compunif\right)\ge 1-\exp\left(-n^{\Theta(\eps)}\right),
\]
as desired.

\subsection{Impossibility for Online Algorithms}

In this section, we prove Theorem~\ref{theo: a_COMP-impossible-Hr}.
Assume for contradiction that there exists an online algorithm $\mathcal{A}$ as defined in Definition~\ref{def : online-arrival} that outputs an independent set
of size at least $(1+\epsilon)\compunif$ with high probability. Our
argument proceeds in three steps:
\begin{enumerate}
    \item Using $\mathcal{A}$, we construct a collection of correlated random instances of $H_r(n, p)$ that remain indistinguishable to $\mathcal{A}$ up to a carefully chosen stopping time.
    \item We show that the algorithm's decisions across these correlated instances yields a tuple of independent sets satisfying certain key structural properties with probability $= \omega(\delta)$ (recall $\delta$ from the statement of Theorem~\ref{theo: a_COMP-impossible-Hr}).
    \item Finally, we prove that such a configuration exists in $H_r(n,p)$ with probability $\ll \delta$, a contradiction.
\end{enumerate}
This contradiction implies that the online model cannot overcome the computational threshold $\compunif$.

Recall from Definition~\ref{def: kdelta-Hr} that our online algorithm uses internal randomness via the vertex arrival order. 
An identical argument to that in~\cite[Section 4.1]{gamarnik2025optimal}, \textit{mutatis mutandis}, implies that it suffices to rule out
\emph{deterministic} online algorithms that $(k,\delta)$-optimize the independent set problem
in $H_r(n,p)$ for the specified values of $k$ and $\delta$. 
Therefore, for the remainder of this section we fix and work with a deterministic online algorithm $\mathcal{A}$.

\subsubsection{Correlated Random Hypergraph Families}\label{subsec:correlated-uniform}
Run $\mathcal A$ on the base hypergraph $H\sim H_r(n, p)$.
For each time $1 \le T \le n$, let $E_{\mathcal A}(T) \subseteq \binom{[n]}{r}$ denote the set of all $r$-hyperedges whose status of presence or absence
is queried by
$\mathcal A$ up to and including time $T$, and let $V_{\mathcal A}(T) \subseteq [n] $ denote the set of vertices exposed by $\mathcal A$ in the first $T$ steps.

Fix $m \ge 1$. For each $T\in\{1,\dots,n\}$, the correlated family
$H_1^{(T)},\dots,H_m^{(T)}$ is constructed so that the entire transcript revealed to
$\mathcal A$ up to time $T$ is identical across all $m$ copies, while all unrevealed hyperedges are resampled independently across $i\in\{1,\dots,m\}$, as follows:
\begin{itemize}
  \item Set $H^{(T)}_1 \coloneqq H$.

  \item For each $i \in \{2,\ldots,m\}$ and each
  $e \in E_{\mathcal A}(T)$, set
  $
  \mathbf 1\{e \in E(H^{(T)}_i)\} \coloneqq \mathbf 1\{e \in E(H)\}.
  $

  \item For each $i \in \{2,\ldots,m\}$ and each
  $e \notin E_{\mathcal A}(T)$, set
  $
  \mathbf 1\{e \in E(H^{(T)}_i)\} \coloneqq X^{(T,i)}_e,
  $
  where
  $
  \left\{X^{(T,i)}_e \,:\, e \notin E_{\mathcal A}(T),\; i=2,\ldots,m \right\}
  $
  is a collection of i.i.d.\ $\mathrm{Ber}(p)$ random variables, independent of $H$.
\end{itemize}
Since $\mathcal A$ is deterministic and all queried hyperedges $e \in E_{\mathcal A}(T)$ have identical status across $H^{(T)}_1,\ldots,H^{(T)}_m$, the behavior of $\mathcal A$ in the first $T$ steps is the \emph{same} on all of $H^{(T)}_1,\ldots,H^{(T)}_m$. Equivalently, for each $i \ge 2$,
\[
E(H^{(T)}_i)
=
\left(E_{\mathcal A}(T)\cap E(H)\right)
\;\cup\;
\left\{ e \notin E_{\mathcal A}(T) : X^{(T,i)}_e = 1 \right\}.
\]

\subsubsection{Forbidden tuples of independent sets}\label{subsec:forbidden-uniform}

Fix $\epsilon>0$ and a deterministic online algorithm $\mathcal{A}$. We will bound the
probability that $\mathcal{A}$ outputs an independent set of size at least
$(1+\epsilon)\compunif$.
Set $\mu\coloneqq\epsilon/2$ and define the stopping time
\[
\tau \coloneqq \min\left\{\, n,\ \min\{T\,:\, |\mathcal{A}(H) \cap V_{\mathcal{A}}(T)|=(1-\mu)\compunif\}\right\}.
\]
We analyze $\mathcal A$ on the correlated family $H_1^{(\tau)},\dots,H_m^{(\tau)}$. We call an $m$-tuple  \((\mathcal A(H_1^{(\tau)}),\dots,\mathcal A(H_m^{(\tau)}))\)
\emph{forbidden} if 
\[
|\mathcal A(H_i^{(\tau)})| \ge (1+\epsilon) \compunif
\qquad\text{for all } i\in[m].
\]
For each $i\in[m]$ and $T \in [n]$, let
\[
\mathcal E_{i,T} \coloneqq \left\{\, |\mathcal A(H_i^{(T)})| \ge (1+\epsilon)\compunif\right\}.
\]
We then define
\[
\mathcal S \coloneqq \bigcap_{1\le i\le m}\mathcal E_{i,\tau},
\]
i.e., $\mathcal S$ is the event that the correlated family forms a forbidden $m$-tuple.
The next propositions provide lower and upper bounds on $\P[\mathcal S]$.
We defer the proofs to Sections~\ref{subsection: prop lb uniform} and~\ref{subsection: prop ub uniform}, respectively.

\begin{proposition}\label{prop:lower-uniform}
Let $\mathcal{E}$ denote the event that $|\mathcal{A}(H)|\ge (1+\epsilon)\compunif$
for $H\sim H_r(n,p)$. Then
\[
\P[\mathcal S]\ \ge\ \P[\mathcal{E}]^{\,m}.
\]
\end{proposition}
\begin{proposition}\label{prop:upper-uniform}
Let $m=C\epsilon^{-2}$ for some sufficiently large constant $C>0$ (depending only on $r$ and $p$).
Then
\[
\P[\mathcal S]\ =\ \exp\left(-\Omega\left((\log_b n)^{\frac{r}{r-1}}\right)\right).
\]
\end{proposition}

With these propositions in hand, we are ready to prove Theorem~\ref{theo: a_COMP-impossible-Hr}.

\begin{proof}[Proof of Theorem~\ref{theo: a_COMP-impossible-Hr}]
Suppose there exists an online algorithm $\mathcal{A}$ that
$((1+\epsilon)\compunif,\delta)$-optimizes the independent set problem in
$H_r(n,p)$. Fix $m=C\epsilon^{-2}$ with $C$ sufficiently large.
By Propositions~\ref{prop:lower-uniform} and~\ref{prop:upper-uniform}, we have
\[
\delta^m \;\le\; \P[\mathcal S]
\;=\;\exp\left(-\Omega\left((\log_b n)^{\frac{r}{r-1}}\right)\right),
\]
implying
\[
\delta \;\le\; \exp\left(-\Theta\left(\epsilon^2(\log_b n)^{\frac{r}{r-1}}\right)\right),
\]
as desired.
\end{proof}

\subsubsection{Lower Bound: Proof of Proposition~\ref{prop:lower-uniform}}\label{subsection: prop lb uniform}

In this section, we will prove Proposition~\ref{prop:lower-uniform} by a careful application of Jensen's inequality and conditional independence. We have
\begin{align*}
    \P[\tau = T,\, \mathcal{S}] &= \E[\P[\tau = T,\,\mathcal{E}_{1, T}, \ldots, \mathcal{E}_{m, T}]\mid E_{\mathcal{A}}(T)] = \E[\ind{\tau = T}\P[\mathcal{E}_{1, T}, \ldots, \mathcal{E}_{m, T}\mid E_{\mathcal{A}}(T)]],
\end{align*}
where we use the fact that $\set{\tau = T}$ is determined by $E_{\mathcal{A}}(T)$ as $\mathcal{A}$ is deterministic.
Note that the algorithm $\mathcal{A}$ is identical for the first $T$ steps on the hypergraphs $H_1^{(T)}, \ldots, H_m^{(T)}$.
With this in hand, we have
\begin{align*}
    \P[\mathcal{E}_{1, T}, \ldots, \mathcal{E}_{m, T}\mid E_{\mathcal{A}}(T)] = \prod_{i = 1}^m\P[\mathcal{E}_{i, T} \mid E_{\mathcal{A}}(T)] = \P[\mathcal{E}_{1, T} \mid E_{\mathcal{A}}(T)]^m,
\end{align*}
where the last step follows since the random variables $X_e^{T, i}$ are i.i.d.
Plugging this into our earlier expression, we have
\begin{align*}
    \P[\mathcal{S}] = \sum_{T = 1}^{n}\P[\tau = T,\, \mathcal{S}]
    &= \sum_{T = 1}^{n}\E[\ind{\tau = T}\P[\mathcal{E}_{1, T}, \ldots, \mathcal{E}_{m, T}\mid E_{\mathcal{A}}(T)]] \\
    &= \sum_{T = 1}^{n}\E[\ind{\tau = T}\P[\mathcal{E}_{1, T}\mid E_{\mathcal{A}}(T)]^m].
\end{align*}
Observing that $\sum_T\ind{\tau = T} = 1$ and $\ind{\tau=i}\ind{\tau=j}=0$ for $i \neq j$, we can further simplify the above to
\begin{align}
    \P[\mathcal{S}] &= \E\left[\left(\sum_{T = 1}^{n}\ind{\tau = T}\P[\mathcal{E}_{1, T}\mid E_{\mathcal{A}}(T)]\right)^m\right] \geq \left(\E\left[\sum_{T = 1}^{n}\ind{\tau = T}\P[\mathcal{E}_{1, T}\mid E_{\mathcal{A}}(T)]\right]\right)^m, \label{eq: success lb}
\end{align}
where we use Jensen's inequality in the final step.
Once again, since $\sum_T\ind{\tau = T} = 1$, we have
\begin{align*}
    \E\left[\sum_{T = 1}^{n}\ind{\tau = T}\P[\mathcal{E}_{1, T}\mid E_{\mathcal{A}}(T)]\right] &= \sum_{T = 1}^{n}\E\left[\ind{\tau = T}\P[\mathcal{E}_{1, T}\mid E_{\mathcal{A}}(T)]\right] = \sum_{T = 1}^{n}\P[\tau = T, \mathcal{E}_{1, T}] = \P[\mathcal{E}].
\end{align*}
Plugging this into \eqref{eq: success lb} completes the proof.

\subsubsection{Upper Bound: Proof of Proposition~\ref{prop:upper-uniform}}\label{subsection: prop ub uniform}

For each $i \in [m]$, we define
\[
a_i \;\coloneqq\; \left|\mathcal A(H_i^{(\tau)})\right|,
\]
denoting the size of the independent sets output by the algorithm on each instance. 

Consider the set of all possible output-size vectors 
\[
\mathcal{F}_m
\;\coloneqq\;
\left\{\vec{a}=(a_1,\dots,a_m)\in [n] ^ m\,:\,  a_i \ge (1+\epsilon)\compunif
\text{ for all } i\right\}.
\]

For each $T \in [n]$ and each $\vec{a} \in \mathcal{F}_m$, let $\mathcal X_{m,T}(a)$ denote the family of $m$-tuples $\left(I_1,\dots,I_m\right)$ satisfying
\begin{itemize}
    \item $I_i \cap V_\mathcal{A}(T) = I_j \cap V_\mathcal{A}(T)$ for $i \neq j$,
    \item $I_i$ is independent in $H_i^{(T)}$, and
    \item $|I_i| = a_i$.
\end{itemize}
Additionally, let
\[
X_{m,T}(\vec{a}) \coloneqq |\mathcal X_{m,T}(\vec{a})|.
\]
On the event $\mathcal S$, the tuple of output sizes
\(
\left(|\mathcal A(H_1^{(\tau)})|,\dots,|\mathcal A(H_m^{(\tau)})|\right)
\)
equals some vector $\vec{a} \in \mathcal{F}_m$, and the corresponding tuple of independent sets belongs to $\mathcal X_{m,\tau}(\vec{a})$.
 Hence,
\begin{equation}\label{eq:S-union}
\mathcal S \;\subseteq\; \bigcup_{\vec{a}\in\mathcal{F}_m}\ \bigcup_{T=1}^{n}\ \left\{X_{m,T}(\vec{a})\ge 1\right\}.
\end{equation}
By a union bound, we have
\begin{equation}\label{eq:union-sum}
\mathbb P[\mathcal S]
\;\le\;
\sum_{\vec{a}\in\mathcal{F}_m}\ \sum_{T=1}^{n}\ \mathbb P\left[X_{m,T}(\vec{a})\ge 1\right].
\end{equation}
Fix $\vec{a}\in\mathcal{F}_m$ and $1\le T\le n$, and set
\[
s_0 \;\coloneqq\; (1-\mu)\compunif.
\]
We define $I\subseteq V_{\mathcal A}(T)$ with $|I|=s_0$ to represent the common portion (among the vertices revealed up to time $T$ by $\mathcal{A}$) that will be shared across the $m$ candidate output
independent sets.
Given such a set $I$, for each $i\in[m]$ define 
\[
\mathcal I_i(I)
\;\coloneqq\;
\left\{
J\subseteq [n]\setminus V_{\mathcal A}(T)\;:\;
|I\cup J|=a_i
\ \text{ and }\ 
I\cup J \text{ is independent in } H_i^{(T)}
\right\}.
\]
Any $(I_1,\dots,I_m)\in\mathcal X_{m,T}(\vec{a})$ determines a unique set $I$ together with $J_i\in\mathcal I_i(I)$, so
\[
X_{m,T}(\vec{a})
\;=\;
\sum_{\substack{I\subseteq V_{\mathcal A}(T)\\ |I|=s_0}}\mathbf{1}\{I \text{ is an independent set in each } H_i^{(T)}\}\prod_{i=1}^m |\mathcal I_i(I)|.
\]
By Markov's inequality,
\begin{align}
&~\mathbb P\left[X_{m,T}(\vec{a})\ge 1\right]
\le\;
\mathbb E[X_{m,T}(\vec{a})] \nonumber \\
&\le\;
\sum_{\substack{I\subseteq V_{\mathcal A}(T)\\ |I|=s_0}}
\ \P[I\text{ is an independent set in } H]\prod_{i=1}^m \mathbb E\left[|\mathcal I_i(I)|\mid I\text{ is an independent set in } H\right]. \label{eq:prob-xmt}
\end{align}
Fix a shared set $I$ and an index $i\in[m]$. For any $J\subseteq [n]\setminus V_{\mathcal A}(T)$ with
$|J|=a_i-s_0$, the set $I\cup J$ is independent in $H_i^{(T)}$ only if every $r$-subset of $I\cup J$
that was not fully revealed by time $T$ is absent. By the construction of the correlated family,
$\mathrm{Ber}(p)$ conditional on $E_{\mathcal A}(T)$, hence each contributes a factor $(1-p)$.  Therefore,
\[
\mathbb P\left[I\cup J \text{ is independent in } H_i^{(T)} \,\big|\, E_{\mathcal A}(T)\right]
\;=\;
(1-p)^{\binom{a_i}{r}-\binom{s_0}{r}}
\;=\;
b^{-\left(\binom{a_i}{r}-\binom{s_0}{r}\right)}.
\]
Consequently,
\[
\mathbb E\left[|\mathcal I_i(I)|\right]
\;\le\;
\binom{n-s_0}{a_i-s_0}\,
b^{-\left(\binom{a_i}{r}-\binom{s_0}{r}\right)}.
\]
Plugging this into \eqref{eq:prob-xmt} gives
\begin{align*}
\mathbb P\left[X_{m,T}(\vec{a})\ge 1\right]
& \le\
\binom{n}{s_0} b^{-\binom{s_0}{r}} \prod_{i=1}^m
\binom{n-s_0}{a_i-s_0}\,
b^{-\left(\binom{a_i}{r}-\binom{s_0}{r}\right)}  \leq b^X,
\end{align*}
where
\begin{align}\label{eq: def X}
    X = (m-1)\left(\binom{s_0}{r} - s_0\log_bn\right) - \sum_{i = 1}^m\left(\binom{a_i}{r} - a_i\log_bn\right).
\end{align}
Let us consider each term above separately.
We have
\begin{align*}
    \binom{s_0}{r} - s_0\log_bn \leq \frac{s_0^r}{r!} - s_0\log_bn 
    &= (\log_bn)^{\frac{r}{r-1}}((r-1)!)^{\frac{1}{r-1}}\left(\frac{(1-\mu)^r}{r} - (1 - \mu)\right) \\
    &\leq (\log_bn)^{\frac{r}{r-1}}((r-1)!)^{\frac{1}{r-1}}\left(\frac{(r-1)\mu^2}{2} - \frac{(r-1)}{r}\right),
\end{align*}
where we use the fact that $(1-\mu)^r \leq 1 - r\mu + \binom{r}{2}\mu^2$ for $\epsilon$ sufficiently small.
Noting that $\binom{x}{r} - x\log_bn$ is increasing in $x$ for $r \geq 2$, we have
\begin{align*}
    \binom{a_i}{r} - a_i\log_bn &\geq \binom{(1+\eps)\compunif}{r} - (1+\eps)\compunif\log_bn \\
    &\geq \frac{(1+\eps)^r(\compunif)^r}{r!} - \Theta((\compunif)^{r-1}) - (1+\eps)\compunif\log_bn \\
    &\geq (\log_bn)^{\frac{r}{r-1}}((r-1)!)^{\frac{1}{r-1}}\left(\frac{(1+\eps)^r}{r} - (1 + \eps) - o(1)\right) \\
    &\geq (\log_bn)^{\frac{r}{r-1}}((r-1)!)^{\frac{1}{r-1}}\left(\frac{(r-1)\eps^2}{2} - \frac{(r-1)}{r} - o(1)\right),
\end{align*}
where we use the fact that $(1+ \eps)^r \geq 1 + r\eps + \binom{r}{2}\eps^2$.
Plugging these estimates into \eqref{eq: def X}, we have that $X$ is at most
\begin{align*}
    &~(\log_bn)^{\frac{r}{r-1}}((r-1)!)^{\frac{1}{r-1}}\left((m-1)\left(\frac{(r-1)\mu^2}{2} - \frac{(r-1)}{r}\right) - m\left(\frac{(r-1)\eps^2}{2} - \frac{(r-1)}{r} - o(1)\right)\right) \\
    &= (\log_bn)^{\frac{r}{r-1}}((r-1)!)^{\frac{1}{r-1}}\left(m\left(\frac{(r-1)\mu^2}{2} - \frac{(r-1)\eps^2}{2} - o(1)\right) - \frac{(r-1)\mu^2}{2} + \frac{(r-1)}{r}\right).
\end{align*}
Plugging in $\mu= \eps/2$ and assuming $m=C\epsilon^{-2}$ for a sufficiently large $C> 0$, we obtain
\[
\mathbb P\left[X_{m,T}(\vec{a})\ge 1\right]\leq \exp(-\Omega((\log{n})^{\frac{r}{r-1}})).
\]
Finally, since $|\mathcal{F}_m|\le n^m$ and $T\le n$, plugging the above bound into \eqref{eq:union-sum} gives
\[
\mathbb P[\mathcal S]
\;\le\;
\sum_{a\in\mathcal{F}_m}\ \sum_{T=1}^{n}\ \mathbb P\left[X_{m,T}(\vec{a})\ge 1\right]
\;\le\;
n\cdot n^m \cdot \exp\left(-\Omega\left((\log_b n)^{\frac{r}{r-1}}\right)\right)
\;=\;
\exp\left(-\Omega\left((\log_b n)^{\frac{r}{r-1}}\right)\right),
\]
for all sufficiently large $n$ (since $m=\Theta(\epsilon^{-2})$).
Together with \eqref{eq:S-union}, this concludes the proof of Proposition~\ref{prop:upper-uniform}.

\section{Balanced Independent Sets in $H(r,n, p)$}

In this section, we prove our results pertaining to $H(r, n, p)$.
We split the section into three subsections devoted to the proofs of Theorems~\ref{theo: gamma_bal a_stat},~\ref{theo: a_comp-achievable-Hrnp}, and~\ref{theo: a_COMP-impossible-Hrnp}, respectively.

\subsection{Statistical Threshold}\label{sec: stat-threshold-rpart}

In this section, we will prove Theorem~\ref{theo: gamma_bal a_stat}.
Let the vertex set of $H$ be
\[
V(H) = V_1 \sqcup  V_2 \sqcup \cdots \sqcup  V_r,
\qquad 
\]
where $\sqcup$ denotes a disjoint union, and $|V_i| = n$ for all $i \in [r]$. Fix a vector $\gamma = (\gamma_1,\dots,\gamma_r) \in \Q^r$ with
$\gamma_i > 0$ and $\sum_{i=1}^r \gamma_i = 1$.
For $\alpha > 0$, let $Z_\alpha(\gamma)$ denote the number of
$\gamma$-balanced independent sets $I$ in $G$ of size $\alpha$.
Since $r! = O(1)$, we may assume without loss of generality that $\sigma$ is the identity permutation (where the permutation $\sigma$ is as in Definition~\ref{definition: balanced independent sets}), i.e.,
\[
|I \cap  V_i| = \gamma_i |I|   = \gamma_i \alpha, \qquad\text{for all } i \in [r].
\]
We first establish the upper bound by a first-moment
argument. For any $\alpha$ we have
\begin{align}\label{1-momment}
    \mathbb{E}[Z_\alpha]
    = \left(\prod_{i=1}^r \binom{n}{\gamma_i \alpha}\right)
      (1-p)^{\alpha^r \prod_{i=1}^r \gamma_i}.
\end{align}
Letting $b \coloneqq 1/(1-p)$, we recall the following parameter from \eqref{eq:Thresholds partite}:
\[
\statpartite
    \coloneqq \left(\frac{\log_b n}{\prod_{i=1}^r \gamma_i}\right)^{1/(r-1)}.
\]
For any $\alpha \ge (1+\epsilon)\statpartite$ we obtain
\begin{align*}
    \mathbb{E}[Z_\alpha]
    \le n^{\alpha} \,
       b^{-\alpha^r \prod_{i=1}^r \gamma_i}
    &= \exp\left(
        \alpha \log n
        - \alpha^r \left(\prod_{i=1}^r \gamma_i\right)\log b
      \right)  \\
      &\leq \exp\left(
        \alpha \left(\log n
        - (1+\epsilon)^{r-1}(\statpartite)^{r-1} \left(\prod_{i=1}^r \gamma_i\right)\log b\right)
      \right) \\
      &=
    \exp\left(
        \alpha \log n \,
        (1 - (1+\epsilon)^{r-1})
    \right). 
\end{align*}
Since $(1+\epsilon)^{r-1} \geq 1 + (r-1)\epsilon$, we have
\[
\mathbb{E}[Z_{\alpha}]
    = \exp(-\Omega\left(\epsilon\, \alpha \log n\right))
    \xrightarrow{n\to\infty}{} 0.
\]
Therefore,
\[
\P(Z_{\alpha} > 0)
    \le \mathbb{E}[Z_{\alpha}] \to 0,
\]
and with high probability there exists no $\gamma$-balanced independent set
of size $\alpha$.

Let us now turn our attention to the lower bound. That is, denoting by $\alpha_\gamma(H)$ the maximum size of a $\gamma$-balanced independent set in $H$, we aim to show that $\alpha_\gamma(H) \geq (1-\epsilon)\statpartite$ with high probability.
Suppose that $\alpha \le (1-\epsilon)\statpartite $, by
the standard binomial coefficient bound,
we obtain
\[
\begin{aligned}
\mathbb{E}[Z_\alpha]
&=
\left(\prod_{i=1}^r \binom{n}{\gamma_i\alpha}\right)(1-p)^{\alpha^r\prod_{i=1}^r\gamma_i} \\
&\ge
\left(\prod_{i=1}^r \left(\frac{n}{\gamma_i\alpha}\right)^{\gamma_i\alpha}\right)(1-p)^{\alpha^r\prod_{i=1}^r\gamma_i} \\
&=
\exp \left(
\alpha\log n-\alpha\log\alpha-\alpha\sum_{i=1}^r\gamma_i\log\gamma_i-\alpha^r\left(\prod_{i=1}^r\gamma_i\right)\log b
\right) \\
&\ge
\exp \left(
\alpha\log n-\alpha(1-\epsilon)^{r-1}\left(\statpartite\right)^{r-1}
\left(\prod_{i=1}^r\gamma_i\right)\log b
-O(\alpha\log\alpha)
\right) \\
&=
\exp \left(
\alpha\log n\bigl(1-(1-\epsilon)^{r-1}\bigr)-O(\alpha\log\alpha)
\right).
\end{aligned}
\]
Therefore,
\[
\mathbb{E}[Z_\alpha] =  \exp(\Omega(\log^{r/(r-1)} n)) = \omega(1).
\]
Since $Z_\alpha \le Z_{\alpha'}$ for all $\alpha \ge \alpha'$, it suffices
to prove $Z_\alpha > 0$ with high probability for $\alpha = \alpha_\epsilon \coloneqq (1-\epsilon)\statpartite$.
Let
\[
Z_{\alpha_\epsilon}
    = \sum_{(I_1,\dots,I_r)}
      I(I_1,\dots,I_r),
\]
where the sum runs over all $(I_1,\dots,I_r)$ satisfying
\[
I_i \subseteq V_i, \qquad |I_i| = \gamma_i \alpha_\epsilon
\qquad \text{for each }i \in [r],
\]
and $I(I_1,\dots,I_r)$ is the indicator that the set $I_1\sqcup \dots \sqcup I_r$ forms an
independent set in $H$.
We next expand
\[
Z_{\alpha_\epsilon}^2
  = \sum_{(I_1,\dots,I_r)}
    \sum_{(I_1',\dots,I_r')}
      I(I_1,\dots,I_r)\,
      I(I_1',\dots,I_r').
\]
To parameterize intersections, define for each $i \in [r]$,
\[
m_i \coloneqq |I_i \cap I_i'|.
\]
These satisfy
\[
0 \le m_i \le \gamma_i \alpha_\epsilon.
\]
Fix a tuple $(I_1,\dots,I_r)$. For an `overlap vector'
\(\vec{m}=(m_1,\dots,m_r), 0\le m_i \le \gamma_i\alpha_\epsilon,\)
the number of tuples $(I_1',\dots,I_r')$ satisfying
\(|I_i\cap I_i'|=m_i \) for each \(i\in[r]\) is at most
\begin{align}\label{eq : N-m}
    N(m_1,\dots,m_r)
    \coloneqq
    \prod_{i=1}^r
    \binom{\gamma_i\alpha_\epsilon}{m_i}
    \binom{n-\gamma_i\alpha_\epsilon}{\gamma_i\alpha_\epsilon-m_i}.
\end{align}
Moreover, for any such pair of tuples, the joint expectation depends only on the overlap vector $\vec{m}$:
\[
\mathbb{E}\big[I(I_1,\dots,I_r)\,I(I_1',\dots,I_r')\big]
=
(1-p)^{2\alpha_\epsilon^r\prod_{i=1}^r\gamma_i-\prod_{i=1}^r m_i} = b^{\,-2\alpha_\epsilon^r\prod_{i=1}^r\gamma_i+\prod_{i=1}^r m_i}.
\]

Therefore,
\begin{align*}
\mathbb{E}[Z_{\alpha_\epsilon}^2]
&=
\sum_{(I_1,\dots,I_r)}
\sum_{(I_1',\dots,I_r')}
\mathbb{E}\big[I(I_1,\dots,I_r)\,I(I_1',\dots,I_r')\big] \\
&\le
\sum_{(I_1,\dots,I_r)}
\sum_{0\le m_i\le \gamma_i\alpha_\epsilon}
N(m_1,\dots,m_r)\,
b^{\,-2\alpha_\epsilon^r\prod_{i=1}^r\gamma_i+\prod_{i=1}^r m_i}.
\end{align*}
Since the inner expression does not depend on the particular choice of $(I_1,\dots,I_r)$, we may factor out the number of such tuples to obtain
\begin{align*}
\mathbb{E}[Z_{\alpha_\epsilon}^2]
&\le
\left(\prod_{i=1}^r \binom{n}{\gamma_i\alpha_\epsilon}\right)
\sum_{0\le m_i\le \gamma_i\alpha_\epsilon}
N(m_1,\dots,m_r)\,
b^{\,-2\alpha_\epsilon^r\prod_{i=1}^r\gamma_i+\prod_{i=1}^r m_i}.
\end{align*}
Hence,  from \eqref{1-momment} and \eqref{eq : N-m}, we have
\begin{align}
\frac{\mathbb{E}[Z_{\alpha_\epsilon}^2]}{\mathbb{E}[Z_{\alpha_\epsilon}]^2}
&\le
\frac{1}{\mathbb{E}[Z_{\alpha_\epsilon}]}
\sum_{0\le m_i\le \gamma_i\alpha_\epsilon}
N(m_1,\dots,m_r)\,
b^{\,-\alpha_\epsilon^r\prod_{i=1}^r\gamma_i+\prod_{i=1}^r m_i} \notag \\
&=
\sum_{0\le m_i\le \gamma_i\alpha_\epsilon}
\frac{N(m_1,\dots,m_r)}{\prod_{i=1}^r \binom{n}{\gamma_i\alpha_\epsilon}}
b^{\prod_{i=1}^r m_i}.\label{eq: 2-momment}
\end{align}
Moreover,
\[
\frac{N(m_1,\dots,m_r)}{\prod_{i=1}^r \binom{n}{\gamma_i\alpha_\epsilon}}
=
\prod_{i=1}^r
\left[
\binom{\gamma_i\alpha_\epsilon}{m_i}
\frac{\binom{n-\gamma_i\alpha_\epsilon}{\gamma_i\alpha_\epsilon-m_i}}{\binom{n}{\gamma_i\alpha_\epsilon}}
\right].
\]
For each \(i\in[r]\), since $0\le m_i \le \gamma_i\alpha_\epsilon \le \alpha_\epsilon \le n$, we have
\begin{align*}
\frac{\binom{n-m_i}{\gamma_i\alpha_\epsilon-m_i}}{\binom{n}{\gamma_i\alpha_\epsilon}}
&=
\frac{(n-m_i)!}{(\gamma_i\alpha_\epsilon-m_i)!\,(n-\gamma_i\alpha_\epsilon)!}
\cdot
\frac{(\gamma_i\alpha_\epsilon)!\,(n-\gamma_i\alpha_\epsilon)!}{n!} \\
& =
\frac{(n-m_i)!\,(\gamma_i\alpha_\epsilon)!}{(\gamma_i\alpha_\epsilon-m_i)!\,n!} \\
& = 
\frac{(\gamma_i\alpha_\epsilon)(\gamma_i\alpha_\epsilon-1)\cdots(\gamma_i\alpha_\epsilon-m_i+1)}
{n(n-1)\cdots(n-m_i+1)} \\ 
& \le 
\left(\frac{\gamma_i\alpha_\epsilon}{n}\right)^{m_i} \le 
\left(\frac{\alpha_\epsilon}{n}\right)^{m_i}.
\end{align*}
Additionally, we have
\[
\binom{\gamma_i\alpha_\epsilon}{m_i}\le (\gamma_i\alpha_\epsilon)^{m_i}\le \alpha_\epsilon^{m_i}.
\]
Therefore,
\[
\frac{N(m_1,\dots,m_r)}{\prod_{i=1}^r \binom{n}{\gamma_i\alpha_\epsilon}}
\le
\prod_{i=1}^r
\alpha_\epsilon^{m_i}
\left(\frac{\alpha_\epsilon}{n}\right)^{m_i}.
\]
Substituting these bounds into \eqref{eq: 2-momment} yields
\begin{align*}
    \frac{\E[Z_{\alpha_\epsilon}^2]}{\E[Z_{\alpha_\epsilon}]^2} &
    \le
    \sum_{0\le m_i\le \gamma_i\alpha_\epsilon}
    \left(\prod_{i=1}^r \alpha_\epsilon^{2m_i} n^{-m_i}\right)
    b^{\prod_{i=1}^r m_i} \\
    & = \sum_{0\le m_i\le\gamma_i \alpha_\epsilon}
     \exp\left(
        \sum_{i=1}^r
            (2m_i\log \alpha_\epsilon - m_i\log n)
        + \left(\prod_{j=1}^r m_j\right)\log b
     \right).
\end{align*}

Recall the definition of the overlap vector $\vec{m}$. We define 
\[
    q(\vec{m}) = \exp\left(
        \sum_{i=1}^r
            (2m_i\log \alpha_\epsilon - m_i\log n)
        + \left(\prod_{j=1}^r m_j\right)\log b
     \right).
\]
To bound $q(\vec{m})$, we distinguish several cases according to $\vec{m}$.

\textbf{Case 1}: If $m_1 = m_2 = \dots = m_r = 0$, then $q(\vec{m}) = 1 $.

\textbf{Case 2}: Suppose there exists $i$ such that $m_i = 0$ and $\sum_im_i > 0$.
Then 
\begin{align*}
q(\vec{m}) 
& \leq \exp \left(-c(\log n - 2 \log \alpha_\epsilon)\right)= \exp\left(-\Omega (\log n)\right), \qquad  \text{for } c = \sum_im_i > 0. \end{align*}

\textbf{Case 3}: Assume $m_1,\dots,m_r \geq 1$ and
$\prod_{i=1}^r m_i \leq \log_b n$. Then,
\begin{align*}
    q(\vec{m}) & \leq \exp\left( -r(\log n - 2 \log \alpha_\epsilon) + \log n \right)  = \exp (- \Omega (\log n)).
\end{align*}

\textbf{Case 4:} Suppose $m_1,\dots,m_r \ge 1$ and
$
\prod_{i=1}^r m_i \ge \log_b n.
$
We can rewrite $q(\vec{m})$ as follows:
\begin{align}
q(\vec{m})
 &= \exp \left(
     \left(\prod_{i=1}^r m_i\right)\log b
     \left(
       1 - \frac{\sum_{i=1}^r m_i}{\prod_{j=1}^r m_j}
           (\log_b n - 2\log_b \alpha_\epsilon)
     \right)
   \right).
\label{eq:q-case4-r}
\end{align}
Since $m_i\ge 1$, $m_i\le \gamma_i\alpha_\epsilon$ and $\sum_{i = 1} ^ r \gamma_i = 1$, we can bound
\[
\frac{\sum_{i=1}^r m_i}{\prod_{i=1}^r m_i}
= \sum_{i=1}^r \frac{1}{\prod_{j\neq i} m_j}
\ge \sum_{i=1}^r \frac{1}{\prod_{j\neq i} (\gamma_j\alpha_\epsilon)}
= \frac{1}{\alpha_\epsilon^{r-1}\prod_{j=1}^r\gamma_j}\sum_{i=1}^r \gamma_i
= \frac{1}{\alpha_\epsilon^{r-1}\prod_{j=1}^r\gamma_j}.
\]
Plugging in $\alpha_\epsilon^{r-1}=(1-\epsilon)^{r-1} (\log_b n) / \left(\prod ^ r _ {i = 1} \gamma_i \right)$, we obtain
\[
\frac{\sum_{i=1}^r m_i}{\prod_{i=1}^r m_i}\ge \frac{1}{(1-\epsilon)^{r-1}\log_b n} = \frac{1+\epsilon'}{\log_b n},
\]
where $\epsilon'= 1/(1 - \epsilon)^{r-1} - 1$ depends only on $\epsilon$ and $r$.
(Note that $\epsilon' \approx (r-1)\epsilon$ for $\epsilon$ sufficiently small.)
Consequently, for sufficiently large $n$ and since $\log \alpha_\epsilon = \Theta(\log\log n)$, we have
\begin{align*}
1 - \frac{\sum_{i=1}^r m_i}{\prod_{j=1}^r m_j}
        (\log_b n - 2\log_b \alpha_\epsilon)
&\le 1 - (1+\epsilon')\left(1 - \frac{2\log_b \alpha_\epsilon}{\log_b n}\right) \le - \frac{\epsilon'}{2}.
\end{align*}
Combining with \eqref{eq:q-case4-r} and using
$\prod_{i=1}^r m_i \ge \log_b n$, we obtain
\[
q(\vec{m})
 \le \exp \left(
         -\frac{\epsilon'}{2}
         \left(\prod_{i=1}^r m_i\right)\log b
      \right)
 \le \exp \left(
         -\frac{\epsilon'}{2}\log_b n \cdot \log b
      \right)
 = \exp\left(-\Omega(\log n)\right).
\]

Combining all four cases for $\vec{m}$, we have
\[
1 \le \frac{\mathbb{E}[Z_{\alpha_\epsilon}^2]}{\mathbb{E}[Z_{\alpha_\epsilon}]^2}
= \sum_{0 \le m_i \le \gamma_i \alpha_\epsilon} q(\vec{m})
\le 1 + \sum_{\substack{0 \le m_i \le \gamma_i \alpha_\epsilon, \\ \sum_im_i > 0}} \exp(-\Omega(\log n)).
\]
Since the number of possible tuples $\vec{m}$ is at most
$\prod_{i=1}^r (\gamma_i \alpha_\epsilon + 1)
 = \exp(O(\log\log n))$, we have
\[
\frac{\mathbb{E}[Z_{\alpha_{\epsilon}}^2]}{\mathbb{E}[Z_{\alpha_{\epsilon}}]^2}
 = 1 + \exp\left(-\Omega(\log n)\right).
\]
Applying the Paley–Zygmund inequality yields
\[
\P(Z_{\alpha_{\epsilon}}>0)
 \;\ge\;
\frac{\mathbb{E}[Z_{\alpha_{\epsilon}}]^2}{\mathbb{E}[Z_{\alpha_{\epsilon}}^2]}
 = 1 - \exp\left(-\Omega(\log n)\right),
\]
which completes the proof of Theorem~\ref{theo: gamma_bal a_stat}.

\begin{remark}
    Analogous to the observation of Remark~\ref{rem:div_r_hr}, Theorem~\ref{theo: gamma_bal a_stat} with the same proof stays true if $r\coloneqq r_n = O\left(\frac{\log \log n}{\log \log \log n}\right)$.
\end{remark}

\subsection{Achievability Result}\label{sec: achiv-rpart}

We describe an online algorithm for constructing a $\gamma$-balanced independent set of input target size $\hat{\alpha}$, where $\gamma \in \Q_+^r$ with $\gamma_1 \le \cdots \le \gamma_r$. The algorithm proceeds in rounds $t = 1, \dots, rn$, where vertices are revealed according to the online arrival model. Let $I_t \subseteq S_t$ denote the current independent set, with $I_0 = \emptyset$. To enforce balancedness, we view the construction as filling $r$ \textit{buckets} sequentially in \textit{stages}. We maintain a set $F_t \subseteq [r]$ of `locked' parts and a bucket index $m_t = |F_t| \in \{0, \dots, r\}$ indicating the current stage. At round $t$, the algorithm accepts $v_t$ only if it belongs to an `unlocked' part and preserves independence; that is, $v_t \in \bigcup_{f \notin F_t} V_f$ and $I_t \cup \{v_t\}$ is independent. If accepted, we set $I_{t+1} = I_t \cup \{v_t\}$; otherwise, $I_{t+1} = I_t$. After this update, if there exists an unlocked part $V_f$ such that
\[
|I_{t+1} \cap V_f| \ge \gamma_{m_t + 1} \hat{\alpha},
\]
then we lock $V_f$ by adding $f$ to $F_t$ and advance to the next stage by setting $m_{t+1} = m_t + 1$. Otherwise, the state remains unchanged. The final output is $I_{rn}$. 
We will show that the algorithm, formally described in Algorithm~\ref{alg: staged and bucketed}, successfully constructs a $\gamma$-balanced independent set of size $\hat{\alpha}$ with high probability as long as $\hat{\alpha} \leq (1-\eps)\comppartite$.

\begin{algorithm}[htb!]
\caption{Greedy Algorithm for Building a $\gamma$-Balanced Independent Set}
\label{alg: staged and bucketed}
$m_1 \gets 0$ ,\qquad  $S_1 \gets \emptyset$ ,\qquad  $I_1 \gets \emptyset$ ,\qquad  $F_1 \gets \emptyset$ 

\For{$t = 1, \ldots, rn$}{
    Sample a vertex $v_t \notin S_t $ \\
    $S_{t+1} \gets S_t \cup \{v_t\}$  \\
    \eIf{$m_t < r$, $v_t \notin \bigcup_{f\in F_t}V_f$, and $I_t\cup\{v_t\}$ is independent}{
        $I_{t+1} \gets I_t \cup \{v_t\}$ 
    }{
        $I_{t+1} \gets I_t$ 
    }
    \eIf{$m_t < r$ \textbf{and} $\max_{f\notin F_t}\,|I_{t+1}\cap V_f| \ge \gamma_{m_t}\hat{\alpha}$}{
        $F_{t+1} \gets F_t \cup \{\arg\max_{f\notin F_t}\ |I_{t+1}\cap V_f|\}$ \\
        $m_{t+1} \gets m_t + 1$
    }{
        $F_{t+1} \gets F_t$ \\
        $m_{t+1} \gets m_t$ 
    }
}
\Return $I_{rn+1}$ 
\end{algorithm}

Define $T_i \coloneqq \inf_{1\leq t \leq rn} \{m_{t} = i\}$ and the events $\tau_i \coloneqq \{T_i \leq rn \}$ for all $0\leq i \leq r$.
Observe that $\tau_0 \supset\tau_1 \cdots \supset\tau_r$ and $\P[\tau_0] = 1$. The probability that the algorithm constructs a $\gamma$-balanced independent set of size $\hat{\alpha}$ is 
\[
\P[\tau_r] = \P[\tau_0]\prod_{i=1}^r \P[\tau_i|\tau_{i-1}] = \prod_{i=1}^r \P[\tau_i|\tau_{i-1}].
\]
In order to show that the probability of failure is $o(1)$, we observe that
\begin{align}\label{eq: partite prob bound}
    \tau_r^c = \tau_0^c \cup (\cup_{i=1}^r (\tau_i^c\cap\tau_{i-1}))
\implies
\P[\tau^c_r] \leq \sum_{i=1}^r \P[\tau_i^c\cap\tau_{i-1}].
\end{align}
We define $E_i \coloneqq \tau^c_i\cap\tau_{i-1}.$ In order to estimate the probability of each of these error terms, let us arbitrarily pick a part $f \notin F_{T_{i-1}+1}$ from the hypergraph. Define 
\[A_f \coloneqq \{v\in V_f\,:\, I_{T_{i-1}+1}\cup\{v\} \text{ is not independent}\}.\]
Observe that 
\[E_i \subseteq \{|A_f| > n-\gamma_i\hat{\alpha}\}\cap\tau_{i-1} \eqqcolon F_i.\]
We will control the probability of $F_i$, and therefore $E_i$, via the first moment method. 
For the event $F_i$ to occur, there must exist a set $S \subseteq V_f$ of size $n - \gamma_i\hat{\alpha}$ such that $I_{T_{i-1}+1}\cup\{v\}$ is not independent for each $v \in S$.
Fix an arbitrary set $S \subseteq V_f$ of size $n - \gamma_i\hat{\alpha}$.
Note that the set $I_{T_{i-1}+1}$ contains exactly $\gamma_j\hat{\alpha}$ vertices from some locked part for each $1 \leq j < i$ and at most $\gamma_i\hat{\alpha}$ vertices from all unlocked parts.
It follows that the probability that $I_{T_{i-1}+1}\cup\{v\}$ is not independent for each $v \in S$ given $I_{T_{i-1}+1}$ is at most
\begin{align*}
    \left(1 - (1-p)^{\hat{\alpha}^{r-1} \gamma_i^{r-i} \prod_{j=1}^{i-1} \gamma_j}\right)^{|S|} &\leq \exp\left(-|S|(1-p)^{\hat{\alpha}^{r-1} \gamma_i^{r-i} \prod_{j=1}^{i-1} \gamma_j}\right) \\
    &\leq \exp\left(-(n - \gamma_i\hat{\alpha})b^{-\hat{\alpha}^{r-1} \gamma_i^{r-i} \prod_{j=1}^{i-1} \gamma_j}\right) \\
    &\leq \exp\left(-(n - \gamma_i\hat{\alpha})b^{-\hat{\alpha}^{r-1} \prod_{j=1}^{r-1} \gamma_j}\right) \\
    &\leq \exp\left(-(n - \gamma_i\hat{\alpha})n^{-(1 - \eps)}\right) \\
    &= \exp\left(-\Omega\left(n^{\eps}\right)\right),
\end{align*}
where we use the fact that $\hat{\alpha} \leq (1-\eps)\comppartite$.
As the above expression is independent of $I_{T_{i-1}+1}$ and since there are at most $\binom{n}{n - \gamma_i\hat{\alpha}} = \exp(O(\log^2n))$ many choices for $S$, we have
\[\P[E_i] \leq \P[F_i] \leq \exp\left(-\Omega\left(n^{\eps}\right)\right).\]
Plugging the above into the expression in \eqref{eq: partite prob bound}, we conclude that the probability of failure is at most $\exp\left(-\Omega\left(n^{\eps}\right)\right)$, as desired.

\subsection{Impossibility for Online Algorithms}\label{subsec:imposs_Hrnp}
We now prove the online impossibility result for $H(r, n, p)$. As in the uniform case, it suffices to consider deterministic online algorithms. To this end, we fix a deterministic online algorithm $\mathcal A$ and run it on a base random hypergraph $H \sim H(r,n,p)$.

\subsubsection{Correlated Random Hypergraph Families}

For each time $T\in[rn]$, let
\begin{itemize}
\item $V_{\mathcal A}(T)$ be the set of vertices exposed by $\mathcal A$ in the first $T$ rounds, and
\item $E_{\mathcal A}(T)$ be the set of all edges whose status has been revealed by $\mathcal A$ up to and including round $T$.
\end{itemize}
For any $m\ge 1$ and any $T\in[rn]$, we define a correlated family of random hypergraphs
\[
H^{(T)}_1,\; H^{(T)}_2,\; \dots,\; H^{(T)}_m
\]
as follows:
\begin{enumerate}
\item $H^{(T)}_1$ is a copy of the base hypergraph $H$.
\item For each $2\le i\le m$ and each edge $e\in E_{\mathcal A}(T)$, the status of $e$ in $H^{(T)}_i$ is forced to be identical to its status in $H$.
\item For each $2\le i\le m$ and each admissible edge $e$ with $e \not\subseteq V_{\mathcal A}(T)$,
      the status of $e$ in $H^{(T)}_i$ is resampled independently:
      \[
      \mathbf 1\{e\in E(H^{(T)}_i)\} \sim \mathrm{Ber}(p)
      \]
    over all such edges $e$, all $T\in[rn]$, and all $i\in\{2,\dots,m\}$.
\end{enumerate}

\subsubsection{Forbidden tuples of independent sets}

Fix $\epsilon > 0$, and set $\mu \coloneqq \epsilon^2$. Let
\[
0 < \gamma_1 \le \cdots \le \gamma_r < 1
\]
be the target proportions, and recall the computational threshold from \eqref{eq:Thresholds partite}:
\[
\comppartite \coloneqq
\left(
\frac{\log_b n}{\prod_{i=1}^{r-1} \gamma_i}\right)^{\frac{1}{r-1}}
=
\left(
\frac{\gamma_r \log_b n}{\prod_{i=1}^r \gamma_i}
\right)^{\frac{1}{r-1}}.
\]

Let $I^{(t)}$ be the partial independent set selected by $\mathcal A$ after $t$ rounds. We define $\tau$ to be the first time by which $I^{(t)}$ contains a $\gamma$-balanced independent set of size $(1-\mu)\comppartite$.
More formally,
\[
\tau \coloneqq
\min\left\{
rn,\;
\min\left\{
t \,:\,
\exists\; \pi \in S_r \text{ s.t. }
|I^{(t)}\cap V_{\pi(j)}|
\geq
(1-\mu)\gamma_j \comppartite
\text{ for each } j\in[r]
\right\}
\right\}.
\]
Here, $S_r$ denotes the symmetric group on $r$ elements.
We let $\pi_\tau$ denote the permutation that witnesses the event in the definition of $\tau$. 

\begin{remark}\label{remark: pi tau}
    We may select $\pi_\tau$ such that $|I^{(\tau)}\cap V_{\pi(j)}|$ does not decrease with $j$. If any $\sigma$ satisfies the condition $|I^{(\tau)}\cap V_{\sigma(j)}| \geq (1-\mu)\gamma_j \comppartite$, then $\pi_\tau$ that ranks $|I^{(\tau)}\cap V_{\pi(j)}|$ in a non-decreasing order satisfies $|I^{(\tau)}\cap V_{\pi_\tau(j)}| \geq (1-\mu)\gamma_j \comppartite$. 
\end{remark}
Fix $m \ge 1$ and consider the correlated graphs $H^{(\tau)}_1,\dots,H^{(\tau)}_m$. For each $i\in[m]$, let
\[
I_i \coloneqq \mathcal A(H^{(\tau)}_i)
\]
be the final output of $\mathcal A$ on $H^{(\tau)}_i$ after $rn$ rounds. Define the successful event
\[
\mathcal S \;\coloneqq\; \bigcap_{i=1}^m \left\{\, |I_i| \ge (1+\epsilon)\,\comppartite \right\}
\;=\; \bigcap_{i=1}^m E_{i,\tau},
\]
where $E_{i,T}\coloneqq\{| \mathcal A(H^{(T)}_i)|\ge (1+\epsilon)\comppartite\}$.
On the event $S$, by construction of the correlated family and the deterministic nature of $\mathcal A$, the following hold.
\begin{itemize}
\item The sets $I_i\cap V_{\mathcal A}(\tau)$ are identical for all $i\in[m]$.
\item There is a `last reached index' \(j^*\in [r]\), and it is met with equality:
\[
|I_i\cap V_{\mathcal A}(\tau)\cap V_{\pi_\tau(j^*)}|
\ =\ (1-\mu)\gamma_{j^*}\, \comppartite
\qquad\text{for all } i\in[m].
\]
\item For each target index $j\in[r]$, the corresponding part $V_{\pi_\tau(j)}$ satisfies
\[
|I_i\cap V_{\mathcal A}(\tau)\cap V_{\pi_\tau(j)}|
\ \ge\ (1-\mu)\gamma_j\, \comppartite
\qquad\text{for all } i\in[m].
\]
\end{itemize}

At this point we state a useful lower bound on $\P[\mathcal S]$.
The proof is identical to the argument in Section~\ref{subsection: prop lb uniform} and so we omit the details.

\begin{proposition}\label{prop:lower-partite}
Let $\mathcal{E}$ denote the event that $|\mathcal{A}(H)|\ge (1+\epsilon)\comppartite$
for $H\sim H(r, n,p)$. Then
\[
\P[\mathcal S]\ \ge\ \P[\mathcal{E}]^{\,m}.
\]
\end{proposition}

Next, we state the key result of this section, which provides an upper bound on $\P[\mathcal S]$.

\begin{proposition}\label{prop: upper-partite}
Let $m=C\epsilon^{-2}$ for some sufficiently large constant $C>0$ (depending only on $r$ and $p$).
Then
\[
\P[\mathcal S]\ =\ \exp\left(-\Omega\left((\log_b n)^{\frac{r}{r-1}}\right)\right).
\]
\end{proposition}

We defer the proof of this proposition to Section~\ref{subsection: proof of ub partite}.
Assuming the result, let us complete the proof of Theorem~\ref{theo: a_COMP-impossible-Hrnp}.

\begin{proof}[Proof of Theorem~\ref{theo: a_COMP-impossible-Hrnp}]
The argument is identical to the proof of Theorem~\ref{theo: a_COMP-impossible-Hr}. Indeed, by Propositions~\ref{prop:lower-partite} and~\ref{prop: upper-partite}, we have
\[
\delta^m \le \mathbb P[S]
\le \exp \left(
-\Omega \left((\log_b n)^{\frac{r}{r-1}}\right)
\right),
\]
for $m=C\epsilon^{-2}$ with $C$ sufficiently large.
It follows that
\[
\delta
\le
\exp \left(
-\Theta \left(\epsilon^2(\log_b n)^{\frac{r}{r-1}}\right)
\right),
\]
as claimed.
\end{proof}

\subsubsection{Upper Bound: Proof of Proposition~\ref{prop: upper-partite}}\label{subsection: proof of ub partite}
For each \(i \in [m]\), we define \(a_i \coloneqq |\mathcal{A}(H_i^{(\tau)})|\) denoting the sizes of independent sets output by the algorithm on each instance.
Note that by Theorem~\ref{theo: gamma_bal a_stat}~\ref{item: z fmm partite}, we have
\[\P\left[\exists i \in [m] \text{ s.t. } |\mathcal{A}(H_i^{(\tau)})| \geq (1+\epsilon)\statpartite\right] = \exp\left(-\Omega\left((\log_bn)^{r/(r-1)}\right)\right)\footnote{\text{The explicit probability bound appears in Section~\ref{sec: stat-threshold-rpart}.}},\]
where we use the fact that $m = \Theta(\epsilon^{-2})$.
It follows that we need only consider $a_i \leq (1+\epsilon)\statpartite$.
Consider the set of all possible output-size vectors 
\[
\mathcal{F}_m\coloneqq\left\{\vec{a}=(a_1,\dots,a_m)\in\mathbb Z^m:\ (1+\epsilon)\comppartite\le a_i\le (1+\epsilon)\statpartite \text { for all }i\right\}.
\]
Note that for the event $\mathcal{S}$ to occur, there exists an $\vec a \in \mathcal{F}$, $T \in [rn]$, permutations $\pi_0, \ldots, \pi_m \in S_r$, and a set of $m$-tuples $(I_1, \ldots, I_m)$ such that the following hold for each $i \in [m]$:
\begin{enumerate}
    \item $|I_i| = a_i$,
    \item $I_i \cap V_{\mathcal{A}}(T) = I_j \cap V_{\mathcal{A}}(T)$ for $i \neq j$,
    \item $I_i$ is a $\gamma$-balanced independent set in $H_i^{(T)}$ with respect to the permutation $\pi_i$, 
    \item $|I_i \cap V_\mathcal{A}(T) \cap V_{\pi_0(j)}| \geq (1- \mu)\gamma_j\comppartite$ for each $j \in [r]$,
    \item $|I_i \cap V_\mathcal{A}(T) \cap V_{\pi_0(j^*)}| = (1- \mu)\gamma_{j^*}\comppartite$ for some $j^* \in [r]$, and
    \item $|I_i \cap V_\mathcal{A}(T)\cap V_{\pi_0(j)}|$ is non-decreasing in $j$(see Remark~\ref{remark: pi tau}).
\end{enumerate}
Let $\chi_{m, T}(\vec{a}, \pi_0, \ldots, \pi_m)$ denote the set of $m$-tuples $(I_1, \ldots, I_m)$ satisfying the above conditions for $\vec a \in \mathcal{F}$, $T \in [rn]$, and permutations $\pi_0, \ldots, \pi_m \in S_r$.
Define the random variable
\[
X_{m, T}(\vec{a}, \pi_0, \ldots, \pi_m)\ \coloneqq\ |\chi_{m, T}(\vec{a}, \pi_0, \ldots, \pi_m)|.
\]
Hence,
\[
\mathcal S\ \subseteq\ \bigcup_{\vec{a}\in\mathcal{F}_m}\ \bigcup_{T=1}^{rn}\ \bigcup_{\pi_0, \ldots, \pi_m \in S_r}
\{X_{m, T}(\vec{a}, \pi_0, \ldots, \pi_m)\ge 1\},
\]
and therefore, by a union bound, we have
\begin{align}\label{eq:choices}
    \mathbb P [\mathcal S]\ \le\ \sum_{\vec{a}\in\mathcal{F}_m}\ \sum_{T=1}^{rn}\ \sum_{\pi_0, \ldots, \pi_m \in S_r}
    \P\left[X_{m, T}(\vec{a}, \pi_0, \ldots, \pi_m)\ge 1\right].
\end{align}
The key result of this section is the following, which provides an upper bound on $\P\left[X_{m, T}(\vec{a}, \pi_0, \ldots, \pi_m)\ge 1\right]$.

\begin{lemma}\label{lem: help-rpart}
For any fixed $\vec a \in \mathcal{F}$, $T \in [rn]$, and permutations $\pi_0, \ldots, \pi_m \in S_r$, we have 
\[
\P\left[X_{m, T}(\vec{a}, \pi_0, \ldots, \pi_m)\ge 1\right] = \exp \left(-\Omega\left(\comppartite\log n\right)\right).
\] 
\end{lemma}

To prove Proposition~\ref{prop: upper-partite}, in addition to proving Lemma~\ref{lem: help-rpart}, it remains to count the number of choices in~\eqref{eq:choices}. These satisfy
\[
|\mathcal{F}_m|\le (rn)^m,\qquad
\#\{T\} = rn,\qquad
\#\{\pi_0, \ldots, \pi_m\} = (r!)^{m+1}. \qquad
\]
Therefore, by the above,~\eqref{eq:choices}, and Lemma~\ref{lem: help-rpart}, we have
\begin{align*}
\P[S]
=
(rn\cdot r!)^{m+1}
\exp \left(
-\Omega \left( \comppartite\log n\right)
\right) =
\exp \left(
-\Omega \left( \comppartite\log n\right)
\right).
\end{align*}
Since
$
 \comppartite
=
\Theta \left((\log n)^{\frac{1}{r-1}}\right)
$,
it follows that
\begin{align*}
\P[S] =
\exp \left(
-\Omega \left((\log n)^{\frac{r}{r-1}}\right)
\right),
\end{align*}
for \(m=C\epsilon^{-2}\), completing the proof of Proposition~\ref{prop: upper-partite} (modulo the proof of Lemma~\ref{lem: help-rpart}, which we include below).

\begin{proof}[Proof of Lemma~\ref{lem: help-rpart}]
    Fix $T$, $\vec{a}$, $\pi_0, \ldots, \pi_m$ as in the statement of the lemma. 
    By the definition of $\chi_{m, T}(\vec{a}, \pi_0, \ldots, \pi_m)$, every $(I_1,\ldots,I_m)\in \chi_{m, T}(\vec{a}, \pi_0, \ldots, \pi_m)$ satisfies
    \[
    I_i\cap  V_{\mathcal{A}}(T)=I_j\cap  V_{\mathcal{A}}(T), \qquad \text{for each } i, j \in [m].
    \]
    Letting $I$ be this common intersection,
    we define the family of admissible intersections to be
    \[
    \mathcal{I}\coloneqq\left\{ I\subseteq  V_{\mathcal{A}}(T) \,:\,
    \begin{aligned}
        &|I| \leq \min_ia_i, \\
        &\forall j\in [r],\, |I\cap V_{\pi_0(j)}|\ge (1-\mu)\gamma_j\comppartite, \\
        &\exists j^* \in [r] \text{ s.t. }
        |I\cap V_{\pi_0(j^*)}|=(1-\mu)\gamma_{j^*}\comppartite    
    \end{aligned}
    \right\}.
    \]
    Fix $I\in\mathcal{I}$. For each $i\in[m]$, define the completion family outside $ V_{\mathcal{A}}(T)$ by
    \[
    \mathcal{I}_i(I)\coloneqq
    \left\{
    J_i\subseteq V(H_i^{(T)})\setminus  V_{\mathcal{A}}(T) \,:\,
    \begin{array}{l}
    I\cup J_i \text{ is } \gamma\text{-balanced with respect to } \pi_i, \\
    |I\cup J_i|=a_i
    \end{array}
    \right\}.
    \]
    For $I\in\mathcal{I}$ and $J_i\in\mathcal{I}_i(I)$, set
    \[
    p(I,J_i)\coloneqq
    \P \left[
    I\cup J_i \text{ is independent in } H_i^{(T)}
    \;\middle|\;
    I \text{ is independent in } H_i^{(T)}
    \right].
    \]
    If $\chi_{m, T}(\vec{a}, \pi_0, \ldots, \pi_m)\ge 1$, then there exists a set $I\in\mathcal{I}$ and `completion sets'
    $J_i\in\mathcal{I}_i(I)$ for $i=1,\ldots,m$ such that $I\cup J_i$ is independent in
    $H_i^{(T)}$ for every $i$.
    Therefore, by a union bound, we have
    \begin{align*}
        \P \left[X_{m, T}(\vec{a}, \pi_0, \ldots, \pi_m)\ge 1\right]
        \le
        \sum_{I\in\mathcal{I}}
        \P \left[
        I \text{ is independent in } H_i^{(T)} \text{ for each } i
        \right]
        \sum_{\substack{J_i\in\mathcal{I}_i(I)\\ 1\le i\le m}}
        p(I, J_i).
    \end{align*}
    Let us fix $I \in \mathcal{I}$ and let $\delta_j \coloneqq |I\cap V_j|$ for each $j \in [r]$.
    Before we proceed, we note a few properties of $\vec\delta = (\delta_1, \ldots, \delta_r)$.
    \begin{itemize}
        \item For some $j^*\in [r]$, we have $\delta_{\pi_0(j^*)} = (1-\mu)\gamma_{j^*}\comppartite$, and
        \item For each $i \in [m]$ and $j \in [r]$, we have
        \begin{align}\label{eq: bounds on delta}
            (1-\mu)\gamma_{\pi_0^{-1}(j)}\comppartite \leq \delta_j \leq \gamma_{\pi_i^{-1}(j)}a_i.
        \end{align}
    \end{itemize}
    Let $\Delta$ denote the set of all such $\vec\delta \in \N^r$ and for a fixed $\vec\delta \in \Delta$, let $\mathcal{I}(\vec\delta)$ be the set of $I \in \mathcal{I}$ such that $|I\cap V_j| = \delta_j$ for each $j$.
    We have
    \begin{align*}
        \P \left[X_{m, T}(\vec{a}, \pi_0, \ldots, \pi_m)\ge 1\right] \nonumber \le
        \sum_{\vec\delta \in \Delta}
        \sum_{I\in\mathcal{I}(\vec\delta)}
        \P \left[
        I \text{ is independent in } H_i^{(T)} \text{ for each } i
        \right]
        \sum_{\substack{J_i\in\mathcal{I}_i(I)\\ 1\le i\le m}}
        p(I, J_i).
    \end{align*}
    Note that
    \[\P \left[
        I \text{ is independent in } H_i^{(T)} \text{ for each } i
        \right] = (1-p)^{\prod_j\delta_j}.\]
    Furthermore, we have
    \begin{align*}
        p(I, J_i) &= (1-p)^{a_i^{r}\prod_j\gamma_j -\prod_j\delta_j}.
    \end{align*}
    Additionally, we note that
    \[|\mathcal{I}(\vec\delta)| = \prod_{j = 1}^r\binom{n}{\delta_j} \leq n^{\sum_{j = 1}^r\delta_j},\]
    and
    \[|\mathcal{I}_i(I)|\leq \prod_{j = 1}^r\binom{n}{\gamma_{\pi_i^{-1}(j)}a_i - \delta_j} \leq n^{a_i - \sum_{j = 1}^r\delta_j}.\]
    Putting these observations together, we obtain
    \begin{align}
        \P \left[X_{m, T}(\vec{a}, \pi_0, \ldots, \pi_m)\ge 1\right]
        &\le
        \sum_{\vec{\delta} \in \Delta}n^{\sum_{j = 1}^r\delta_j}(1-p)^{\sum_{i = 1}^ma_i^{r}\prod_j\gamma_j-(m - 1)\prod_j\delta_j}n^{\sum_{i = 1}^ma_i - m\sum_{j = 1}^r\delta_j} \nonumber \\
        &= \sum_{\vec{\delta} \in \Delta}b^{(m - 1)(\prod_j\delta_j - \log_bn\sum_j\delta_j) - \sum_{i = 1}^ma_i^{r}\prod_j\gamma_j + \log_bn\sum_{i = 1}^ma_i} \label{eq: fub}
    \end{align}
    Denote the exponent above by $X(\vec\delta)$.
    Letting
    \begin{align}\label{eq: prime def}
        \delta_j' = \frac{\delta_j}{\gamma_{\pi_0^{-1}(j)}\comppartite} \qquad \text{and} \qquad a_i' = \frac{a_i}{\comppartite},
    \end{align}
    for $j \in [r]$ and $i \in [m]$, we have
    \[X(\vec\delta) = \comppartite\log_bn\left((m - 1)\left(\gamma_r\prod_j\delta_j'  - \sum_j(\gamma_{\pi_0^{-1}(j)}\delta_j')\right) - \sum_{i = 1}^m\left(\gamma_r(a_i')^{r} - a_i'\right)\right).\]
    To simplify the above, we define
    \[\epst \coloneqq \min_{i \in [m]}\left(\frac{a_i}{\comppartite} - 1\right).\]
    Note that the function $f(x) = \gamma_rx^r - x$ is increasing for $x^{r-1} \geq 1/(\gamma_rr)$.
    As $a_i' \geq (1+\epst)$ and $\gamma_r \geq 1/r$, this condition is satisfied for $x = a_i'$.
    In particular, we have
    \[\gamma_r(a_i')^{r} - a_i' \geq (1+\epst)^{r}\gamma_r - (1+\epst).\]
    Additionally, by relabeling indices, we can assume that $\pi_0$ is the identity permutation.
    It follows that
    \begin{align}\label{eq: bad boi}
        X(\vec\delta) \leq \comppartite\log_bn\left((m - 1)\left(\gamma_r\prod_j\delta_j'  - \sum_j(\gamma_{j}\delta_j')\right) - m\left((1+\epst)^{r}\gamma_r - (1+\epst)\right)\right).
    \end{align}

    In the remainder of this proof, we will focus on providing a bound for the above. To this end, let 
    \[f(\vec{\delta}) = \gamma_r\prod_j\delta_j'  - \sum_j(\gamma_{j}\delta_j').\]
    where we recall the definition of $\delta_j'$ from \eqref{eq: prime def}.
    The following claim will assist with our proof.

    \begin{claim}
        Without loss of generality, we can assume that for each $j \in [r]$, we have $1 - \mu \leq \delta_j' \leq 1 + \tilde \eps$.
    \end{claim}

    \begin{claimproof}
        The lower bound follows by the definition of the stopping time.
        Assume that $\delta_k' > 1+\epst$ for some $k \in [r]$. Recall that by convention, $\delta_{j}$ is increasing (see Remark~\ref{remark: pi tau}). By definition of $\epst$, we have the following for some $1\leq i\leq m$:
        $$
        \delta_j'\gamma_j\comppartite = \delta_j > \gamma_k a_i \qquad \forall j\geq k.
        $$
        This would imply that there are at most $k-1$ indices $j$ such that $\delta_j \leq \gamma_k a_i$. Therefore, any $\gamma$-balanced independent set containing $I^{(T)}$ must be larger than $a_i$, indicating that $X_{m, T}(\vec{a}, \pi_0, \ldots, \pi_m) = 0$.     
    \end{claimproof}
    \vspace{5pt}

    In light of the above claim, we may simplify \eqref{eq: fub} to
    \begin{align*}
        \P \left[X_{m, T}(\vec{a}, \pi_0, \ldots, \pi_m)\ge 1\right]
        &\le \sum_{\vec{\delta} \in \Delta'}b^{X(\vec\delta)},
    \end{align*}
    where
    $$
    \Delta' \coloneqq \{\vec{\delta}\in \Delta: 1- \mu \leq \delta_j'\leq 1+\epst\}.
    $$

Let $f^* =\max_{\vec{\delta} \in \Delta'}f(\vec{\delta})$.
The following claim shows that there is a maximizer which is an extreme point of $\Delta'$.

\begin{claim}
The value $f^*$
on $\Delta'$ is attained at a point of
\[
B \;=\; \Bigl\{\vec{\delta}\in\Delta' \,: \,\delta_j\in\{1-\mu,\,1+\epst\} \qquad \forall j\in[r]\Bigr\}.
\]
\end{claim}

\begin{claimproof}
Write
\[
F_{j^*} \;\coloneqq\; \Bigl\{\vec{\delta}\in \Delta' \,:\, \delta_{j^*}'=1-\mu\Bigr\},
\qquad j^*\in[r],
\]
and note that by definition
\[
\Delta' \;=\; \bigcup_{j^*\in[r]} F_{j^*}.
\]
It suffices to show that, for every $j^*\in[r]$, the maximum of $f$
on $F_{j^*}$ is attained at an extreme point of $F_{j^*}$, since the union of these 
points is precisely $B$.
For any $j\in[r]$ and any choice of $\delta_k$, $k\ne j$, the map
\[
t \;\longmapsto\; f(\delta_1,\dots,\delta_{j-1},t,\delta_{j+1},\dots,\delta_r)
\;=\; \alpha_j\, t + \beta_j,
\qquad
\text{where}
\qquad
\alpha_j \;=\; \gamma_r\!\!\prod_{k\neq j}\delta_k' - \gamma_j,
\]
is affine in $t$, with $\beta_j$ independent of $t$. Thus $f$ is multilinear. Fix $j^*\in[r]$ and let $\vec{\delta}^\star\in F_{j^*}$ be any maximizer of $f$ on $F_{j^*}$. We show that $\vec{\delta}^\star$ may be taken to be an extreme point of $F_{j^*}$. Suppose some coordinate $\delta^\star_j$ with $j\neq j^*$ satisfies $(\delta^\star_j)' \in (1-\mu,1+\epst)$. By the multilinearity of $f$, the function
\[
\varphi(t) \;\coloneqq\; f(\delta^\star_1,\dots,\delta^\star_{j-1},\,t,\,\delta^\star_{j+1},\dots,\delta^\star_r)
\]
is affine on $[1-\mu,1+\epst]$. An affine function on a closed interval attains its maximum at an endpoint, so that
\[
\max\bigl\{\varphi(1-\mu),\;\varphi(1+\epst)\bigr\} \;\ge\; \varphi(\delta^\star_j).
\]
Replacing $\delta^\star_j$ with the corresponding endpoint produces another maximizer in $F_{j^*}$ with one more coordinate at the boundary $\{1-\mu,1+\epst\}$. Iterating this procedure over all $j\neq j^*$ yields a maximizer $\widetilde{\vec{\delta}}\in F_{j^*}$ with
\[
\widetilde\delta_{j^*}'=1-\mu,\qquad 
\widetilde\delta_j'\in\{1-\mu,1+\epst\}\qquad \forall j\neq j^*,
\]
that is, an extreme point of $F_{j^*}$. Such a point must lie in $B$. Taking the maximum over $j^*\in[r]$ gives a global maximizer in $B$, as desired.
\end{claimproof} 
\vspace{5pt}

From now on, we may assume that $f^* = f(\vec{\delta}^*)$ with $\vec{\delta}^*\in B$. Define \[A \coloneqq \{j\in [r]: \delta^*_j = 1-\mu\},\]
and let $s = |A|$. Note that we have $s\geq 1$ by definition of $\Delta'$. Let 
\[\sigma_A \coloneqq \sum_{j\in A}\gamma_j. \] From now on, we substitute $\mu = \epsilon^2 \leq \epst^2$. 
Recall the bound on $X(\vec\delta)$ from \eqref{eq: bad boi}.
We have
\begin{align}
    \frac{X(\vec\delta)}{\comppartite\log_bn} &\leq (m - 1)f(\vec\delta)- m\left((1+\epst)^{r}\gamma_r - (1+\epst)\right) \nonumber \\
    &= m\left(f(\vec\delta)- (1+\epst)^{r}\gamma_r + (1+\epst)\right) - f(\vec{\delta}).\label{eq: bad boi 2}
\end{align}
Let is consider the term in the parentheses.
Recalling the optimizer $\vec{\delta}^*$ of $f(\vec\delta)$, we have
\begin{align*}
    &~f(\vec\delta)- (1+\epst)^{r}\gamma_r + (1+\epst) \\
    &\leq \gamma_r\prod_j(\delta_j^*)' - \sum_j(\gamma_j(\delta_j^*)') - (1+\epst)^{r}\gamma_r + (1+\epst) \\
    &= \gamma_r(1+\epst)^{r-s}(1-\mu)^s - \sigma_A(1-\mu)-(1-\sigma_A)(1+\epst) - (1+\epst)^{r}\gamma_r + (1+\epst) \\
    &\leq \gamma_r(1+\epst)^{r}\left( \left(\frac{1-\mu}{1+\epst} \right)^s-1\right)+\sigma_A(\mu+\epst) \\
    &\le
\gamma_r(1+\epst)^r
\left(
\left(\frac{1-\mu}{1+\epst}\right)^s-1
\right)
+s\gamma_r(\mu+\epst)
\eqqcolon g_A(\epst),
\end{align*}
where we use the fact that since $\gamma_j\le \gamma_r$ for all $j\in[r]$, we have
$\sigma_A\le s\gamma_r$. 
We will show that $g_{A}(\epst)$ is a decreasing function of $\epst$ as
$$
g_A'(\epst) = \gamma_r(1+\epst)^{r-1}\left( (r-s)\left(\frac{1-\mu}{1+\epst} \right)^s-r\right)+\sigma_A \leq -s\gamma_r(1+\epst)^{r-1}+ s\gamma_r < 0.
$$
Therefore, 
\begin{align}
f(\vec{\delta}) - (1+\epst)^{r}\gamma_r + (1+\epst) &\leq g_A(\epst) \leq g_A(\epsilon) = \gamma_r(1+\epsilon)^{r} \left((1-\epsilon)^s-1\right)+\sigma_A\epsilon(\epsilon+1) \nonumber \\
&\leq s\gamma_r \epsilon(1+\epsilon)\left((1+\epsilon)^{r-1}\left(-1+\frac{(s-1)}{2}\epsilon\right) +1 \right) \nonumber \\
& \leq - \frac{r-1}{4}\epsilon^2 s\gamma_r(1+\epsilon)\leq -\frac{\epsilon^2}{8}, \label{eq: bad boi 3}
\end{align}
where we use the second order Taylor approximation for $(1-\epsilon)^s$, a first order approximation for $(1+\epsilon)^{r-1}$, and the bounds $s \leq r$ and we assume that $\epsilon$ is sufficiently small (e.g. $\epsilon \leq 1/r^2$).
Meanwhile, we also have
\begin{align}\label{eq: almost there}
    -f(\vec\delta) = \sum_{j=1}^r \gamma_j \delta_j' - \gamma_r\prod_{j-1}^r \delta_j' \leq 1+\epst \leq \left(\frac{1}{\gamma_r}\right)^\frac{1}{r-1}\leq 2,
\end{align}
since we assume $a_i \leq (1+\epsilon)\statpartite$ for all $i$.
Combining \eqref{eq: bad boi 2}, \eqref{eq: bad boi 3}, and \eqref{eq: almost there} for $m = C\eps^{-2}$ with $C$ sufficiently large, we have
\begin{align*}
    \P \left[X_{m, T}(\vec{a}, \pi_0, \ldots, \pi_m) \ge 1\right] \leq b^{X(\vec{\delta})} &= \exp\left(-\comppartite\log n(m\eps^2/8 - 2)\right) \\
    &= \exp\left(-\Omega\left(\comppartite\log n\right)\right),
\end{align*}
as desired.
\end{proof}

\printbibliography 

\end{document}